\documentclass[aps,prx,superscriptaddress,reprint,floatfix]{revtex4-1}
\pdfoutput=1

\usepackage[utf8]{inputenc}
\usepackage[english]{babel}
\usepackage[T1]{fontenc}
\usepackage{amsmath}
\usepackage{hyperref}

\usepackage{tikz}
\usepackage{lipsum}

\usepackage{graphicx}
\usepackage{grffile}
\usepackage{amsthm}
\usepackage{epsfig}
\usepackage{amsmath,amssymb,amsfonts,mathtools,bbm,mathrsfs,xfrac
}
\usepackage{longtable}
\usepackage{tabularx}
\usepackage{colortbl}

\usepackage{soul}
\usepackage{lipsum}

\renewcommand{\[}{\begin{equation}}
\renewcommand{\]}{\end{equation}}

\newcommand{\ket}[1]{|#1\rangle}
\newcommand{\bra}[1]{\langle#1|}
\newcommand{\braket}[2]{\langle#1|#2\rangle}
\newcommand{\pro}[2]{|#1\rangle\langle#2|}

\newcommand{\tr}{\mathrm{tr}}

\newcommand{\norm}[1]{\left\lvert\left\lvert#1\right\rvert\right\rvert}

\newcommand{\R}{{\hat{\rho}}}
\renewcommand{\S}{{\hat{\sigma}}}
\newcommand{\I}{{\hat{I}}}

\newcommand{\C}{{\mathcal{C}}}
\newcommand{\Cvn}{{\C}}
\newcommand{\Cp}{{\C_{B_\mathrm{projective}}}}
\newcommand{\bC}{{\boldsymbol{\mathcal{C}}}}
\renewcommand{\P}{\hat{P}}
\newcommand{\K}{\hat{K}}

\newcommand{\bi}{{\boldsymbol{i}}}
\newcommand{\bj}{{\boldsymbol{j}}}
\renewcommand{\bm}{{\boldsymbol{m}}}

\newcommand{\hPi}{{\hat{\Pi}}}
\newcommand{\PO}{{\mathcal{A}}}
\newcommand{\HS}{\mathcal{H}}

\newcommand{\U}{{\hat{U}}}

\newcommand{\bx}{{\boldsymbol{x}}}

\definecolor{mygray}{gray}{0.6}

\theoremstyle{definition}
\newtheorem{definition}{Definition}
\newtheorem{theorem}{Theorem}
\newtheorem{lemma}{Lemma}

\definecolor{dfcol}{cmyk}{1, 0.2108, 0.13, 0.3}
\newcommand{\df}[1]{\ifthenelse{\boolean{}}{\textcolor{dfcol}{[{\bf DF}: #1]}}{}}

\begin{document}

% title
\title{Quantifying Information Extraction using Generalized Quantum Measurements
}

% authors

\author{Dominik \v{S}afr\'{a}nek}
\email{dsafranekibs@gmail.com}
\affiliation{Center for Theoretical Physics of Complex Systems, Institute for Basic Science (IBS), Daejeon 34126, Republic of Korea}
\affiliation{SCIPP and Department of Physics, University of California, Santa Cruz, California 95064, USA}

\author{Juzar Thingna}
\email{jythingna@gmail.com}
\affiliation{Center for Theoretical Physics of Complex Systems, Institute for Basic Science (IBS), Daejeon 34126, Republic of Korea}
\affiliation{Basic Science Program, University of Science and Technology (UST), Daejeon 34113, Republic of Korea}

\date{07/03/2022}%Need to do manually

\begin{abstract}
Observational entropy is interpreted as the uncertainty an observer making measurements associates with a system. So far, properties that make such an interpretation possible rely on the assumption of ideal projective measurements. We show that the same properties hold even when considering generalized measurements. Thus, the interpretation still holds: Observational entropy is a well-defined quantifier determining how influential a given series of measurements is in information extraction. This generalized framework allows for the study of the performance of indirect measurement schemes, which are those using a probe. Using this framework, we first analyze the limitations of a finite-dimensional probe. Then we study several scenarios of the von Neumann measurement scheme, in which the probe is a classical particle characterized by its position. Finally, we discuss observational entropy as a tool for quantum state inference. Further developed, this framework could find applications in quantum information processing. For example, it could help in determining the best read-out procedures from quantum memories and to provide adaptive measurement strategies alternative to quantum state tomography.
\end{abstract}
\maketitle

\section{Introduction}
Entropy is a fundamental concept spanning diverse fields such as thermodynamics~\cite{parrondo2015thermodynamics}, machine learning~\cite{rubinstein2013cross}, cosmology \cite{Almheiri2021}, and many more. It is also one of the most misunderstood concepts due to the abundance of definitions. Some definitions come from \emph{thermodynamics}, such as Clausius-~\cite{clausius1867the}, microcanonical- (surface)~\cite{dunkel2014consistent,hilbert2014thermodynamic,buonsante2016dispute}, and Gibbs-entropy~\cite{jaynes1965gibbs,goldstein2019gibbs}. Whereas, some from \emph{information-theoretic} perspective, such as Shannon-~\cite{shannon1948mathematical}, von Neumann-~\cite{von1927thermodynamik}, or entanglement-entropy~\cite{wiseman2003entanglement, bennett2011postulates,plenio2014introduction,girolami2017quantifying}. Among these, various permutations can be made by using different types of \emph{coarse-grainings}, such as Boltzmann~\cite{boltzmann2003further,goldstein2004boltzmann,goldstein2019gibbs}, coarse-grained~\cite{von2010proof,wehrl1978general,gemmer2014entropy,safranek2019short,safranek2019long,engelhardt2019coarse,Almheiri2021,safranek2021brief}, or Kolmogorov-Sinai-entropy~\cite{farmer1982information,Latora1999kolmogorov,frigg2004sense,jost2006dynamical}. 

These notions often overlap under certain circumstances. For instance, Boltzmann entropy defined for a general coarse-graining reduces to microcanonical entropy for energy coarse-graining. %\footnote{Other examples are: i) Fundamental thermodynamic relation states that for systems in thermal equilibrium, microcanonical entropy equals Clausius entropy; ii) For typical systems---those unbounded in energy---microcanonical entropy equals Gibbs entropy in the thermodynamic limit~\cite{touchette2015equivalence}; iii) Von Neumann entropy reduces to Gibbs entropy when taking a thermal ensemble; iv) Entanglement entropy of a system in equilibrium is equal to the Gibbs entropy of a subsystem~\cite{deutsch2013microscopic,popescu2006entanglement,yamaguchi2019theorems}, and so on.}.
In general, since each definition corresponds to a different physical scenario, it needs to be carefully applied to the problem at hand. Despite these subtleties, there is one common denominator to all of these entropies: An increase in every type of entropy represents loss of a \emph{specific} type of information, and a decrease represents an information gain~\footnote{Thermodynamic entropies encode information about the energy eigenstate (microstate) of the system due to an increase in energy. Whereas, information-theoretic entropies capture details about the state of the system due to its interaction with an external system.}. %iii) Coarse-graining-based entropies measure loss of information due to the state of the system evolving so that they wander into larger macrostates, whose inner structure is inaccessible by a macroscopic measurement. Thus knowledge of the state of the system is being effectively lost.}.

In this work, we will focus on a coarse-graining-based type of entropy, called \emph{observational entropy} [see Append.~\ref{app:History} for a historical overview]. %It has been successfully applied to local particle number and energy coarse-grainings and is shown to grow in nonequilibrium
It has been studied in relation to black holes~\cite{Almheiri2021,schindler2021unitarity}, big bang~\cite{lautaro2021quantum}, high-resolution measurement setups~\cite{uzdin2021passivity}, correlated finite baths~\cite{riera2021quantum}, nonequilibrium entropy growth~\cite{safranek2019long,safranek2019classical,safranek2021brief},
entanglement entropy~\cite{faiez2020typical}, 
historically standard entropies~\cite{goldstein2019gibbs}, and various other scenarios~\cite{gemmer2014entropy,hu2019generalized,anza2017information,lent2019quantum,yoshida2020boltzmann,schmidt2020framework}. Moreover, in open quantum systems, it has also been used to systematically bind the first and second law of quantum thermodynamics~\cite{strasberg2021first}. All of these works consider only projective coarse-grainings, i.e., those given by a complete set of projectors, which represent a projective measurement. Projective measurements are a natural first step, since they can be cleanly connected to macrostates. This is simply because each projector corresponds to a Hilbert space subspace and defines a set of states that give the same value of a macroscopic observable, i.e., a macrostate.% It is mainly due to this connection that the Definition~\eqref{eq:vNOE} connects well with Boltzmann's original notion of entropy. However, if we look at the coarse-graining more abstractly, as a description of a macroscopic measurement to be performed (which is not necessarily tied to subspaces but rather to measurement outcomes), the Definition~\eqref{eq:vNOE} could be generalized further. 

In this work, we aim to go beyond the restraints of projective measurements and extend the properties of observational entropy to generalized measurements given by a set of quantum operations, called positive operator-valued measure (POVM)\footnote{Here, we are connecting with well-known notion of POVM at the expense of being slightly imprecise. A POVM can uniquely identify the probability of a measurement outcome given the density matrix, but it does not identify the resulting state after the measurement. The generalized measurement, which we will describe in the next section, can do both, and thus also defines the corresponding POVM.}.
In these generalized measurements, the measurement probe (aka a pointer) can briefly interact with the quantum system collecting information about the measured observable. The pointer stores this information and can be read out at any time without the need to measure the system directly. Within such generalized measurements, the measured observable need not commute at different times or with the system Hamiltonian (unlike quantum nondemolition measurement~\cite{braginsky1980quantum,nakajima2019quantum}). Also, such generalized measurements cannot always be described by rank-1 POVMs~\cite{kraus1983states}. Moreover, the pointer can even be reused to store the sum of values of an observable~\cite{thingna2020}, without any constraints on how the system dynamically interacts with the pointer.

The notion of such generalized measurements is not new and dates back to original works by von Neumann~\cite{von1955mathematical}. The idea has been implemented to study unambiguous quantum state discrimination~\cite{huttner1996unambigious,clarke2001experimental,biggerstaff2009experiments,becerra2013implementation,zhao2015experimental}, measuring non-commuting observables~\cite{chen2019experimental}, certification and device independence testing~\cite{,tavakoli2020self,smania2020experimental}, post-selection of experimental outcomes~\cite{lu2014experimental}, and information-theoretic uncertainty relations~\cite{sponar2021experimental}. %The original idea has also been pushed theoretically by showing that a quantum work measurement is a generalized one~\cite{roncaglia2014nonequilibrium}, or by proposing new protocols to perform generalized measurements using only classical resources and an auxiliary qubit~\cite{singal2021implementation}.

Given the generality of these measurements and their experimental relevance, we re-introduce observational entropy within the framework of generalized measurements. The resulting observational entropy measures the remaining uncertainty about the initial state of the system that was not uncovered by a sequence of given (possibly non-commuting) generalized measurements. In other words, its value measures how much information these measurements extract -- the smaller the value, the more extracted information. Since this entropy can be calculated for different coarse-grainings, representing different measurements, its value serves as a performance quantifier of different (and completely general) measurement schemes in information extraction. 
%There were several obstacles in defining the correct figures of merit (especially the notion of coarse-graining and macrostate volumes), due to not being able to depend on the properties of projectors---their orthogonality in particular. In spite of this, it is possible to define them in such a way that all of the important properties of the resulting entropy are preserved, even for this general case. %
Overall, our work provides a fundamental backbone to information extraction: i) it presents a quantifier for how much information we can extract, justified by the presented theorems, and ii) it is both computable and measurable by experimentalists performing any generalized measurements on quantum systems. %When certain criteria as elucidated in Sec.~\ref{sec:quantum_tomography} for this quantity are met, one can use the presented method to identify the quantum state of the system.

\subsection{Contributions and assumptions}

The paper is organized as follows: In Sec.~\ref{sec:properties}, we first review the known results on observational entropy with projective coarse-grainings. Then we extend this concept to generalized measurements and show that with the new definition, the corresponding three fundamental theorems on observational entropy follow. These confirm the interpretation of this quantity as a measure of observers' ignorance given that they can perform several, but not all types of measurements. This represents our main result. The measurements we consider are the most general ones: thus, no further generalization of observational entropy with regard to a type of measurement is possible. Section~\ref{sec:app_indirect_measurement} considers a the scenario of an indirect measurement scheme. Specifically, it aims to answer how well an indirect measurement performs as compared to a direct measurement, based on general considerations such as the dimension of our auxiliary system (pointer/probe). Section~\ref{sec:von_neumann} shows a specific example of an indirect measurement scheme called the von Neumann measurement scheme. In this scheme, the probe is a massive particle whose position is measured. In Sec.~\ref{sec:numerics}, we perform a number of numerical simulations of this scheme and illustrate how our general theory can be applied to \emph{quantitatively} describe which measurement strategy outperforms others in information extraction. In Sec.~\ref{sec:quantum_tomography}, we show that if a measurement is found such that its corresponding observational entropy is equal to the von Neumann entropy, it is possible to reconstruct the quantum state of the system. Thus, we provide an explicit algorithm for quantum state inference. This opens an exciting new direction, potentially providing an alternative to quantum tomography. Finally, in Sec.~\ref{sec:conclusion}, we summarize our results and provide possible future directions and applications.

% In this work, we will be concerned with the development of Boltzmann entropy and its generalizations. 

\section{Observational entropy with generalized measurements: Bounds on extracted information}\label{sec:properties}

In this section, we first review previously found analytical properties of observational entropy. These properties inspired the interpretation of observational entropy as a measure of the observer's uncertainty about a quantum system. Until now, however, observational entropy was defined only for coarse-grainings given by a projective measurement. To see whether this interpretation holds in general, we extend this notion to include completely general measurements. As the main result of this section, we prove that all properties of the original definition translate also to this general case. This shows that its interpretation as a measure of observer's uncertainty is still valid. Also, these results will illuminate the way on how to use this quantity in its full generality.

\subsection{Observational entropy with projective measurements: Review}

First, we review the basic ideas governing observational entropy for projective coarse grainings that have been extensively discussed in Refs.~\cite{safranek2019classical, safranek2019short, safranek2019long, safranek2021brief}.

Defining $\P_i$ as the projector onto a subspace $\HS_i$, we collect these projectors into a set $\C=\{\P_i\}$, which we call a \emph{coarse-graining}. Projectors in this set are Hermitian ($\P_i^\dag=\P_i$), orthogonal ($\P_i \P_j = \P_i \delta_{ij}$), and they satisfy the completeness relation ($\sum_i\P_i=\hat{I}$). Conversely, any coarse-graining $\C=\{\P_i\}$ with the above properties defines a decomposition of the Hilbert space $\HS=\bigoplus_i\HS_i$, in which a subspace $\HS_i$ (macrostate) is spanned by eigenvectors of $\P_i$. Thus, in this construction, we can either decompose the unity operator into orthonormal projectors ($\hat{I}=\sum_i\P_i$), or equivalently decompose the Hilbert space into subspaces ($\HS=\bigoplus_i\HS_i$).

Each observable defines a coarse-graining $\C_{\hat{A}}=\{\P_a\}$ through its spectral decomposition $\hat{A}=\sum_aa \P_a$. Thus it also defines the corresponding decomposition of the Hilbert space $\HS=\bigoplus_a\HS_a$. Here, each eigenvalue $a$ represents the macroscopic property and is one of the measurement outcomes obtained when measuring that observable. Generally, any coarse-graining can be viewed as a type of measurement. It does not need to be complete, i.e., it does not have to project onto a pure state (which corresponds to rank-1 projectors). Instead, we allow projectors $\P_a$ to have an arbitrary rank. By definition, this rank is the same as the dimension of the subspace $\HS_a$ the projector projects onto.

Given a single coarse-graining, \emph{observational entropy} (also known as coarse-grained entropy~\cite{von1955mathematical,wehrl1978general,Almheiri2021}) is defined as~\cite{safranek2019short,safranek2019long,safranek2019classical,faiez2020typical,schindler2020entanglement,strasberg2021first},
\[
\label{eq:def1}
S_{\C}\equiv-\sum_{i}p_i\ln \frac{p_i}{V_i},
\]
where $p_i=\tr[\P_i\R]$ denotes the probability of finding the state in macrostate $\HS_i$. The volume of that macrostate $V_i=\tr[\P_i]=\dim \HS_i$ is the number of orthogonal states in it. Equivalently, we can call $p_i$ the probability of obtaining a measurement outcome $i$, when measuring in a basis given by coarse-graining $\C$.

This can be generalized to multiple coarse-grainings, for which observational entropy is defined as~\cite{safranek2019long},
\[
\label{eq:def2}
S_{\C_1,\dots,\C_n}\equiv-\sum_{\bi}p_\bi\ln \frac{p_\bi}{V_\bi}.
\]
Above, $\bi=(i_1,\dots,i_n)$ is a vector of outcomes (properties of the system), $p_\bi=\tr[\P_{i_n}\cdots\P_{i_1}\R\P_{i_1}\cdots\P_{i_n}]$ is the probability of obtaining these outcomes in the given order, $V_\bi=\tr[\P_{i_n}\cdots\P_{i_1}\cdots\P_{i_n}]$ is an (ordered) volume of the corresponding macrostates. Note that here, the equivalence between the projectors and subspaces has already been lost: volume $V_\bi$ can be a fraction, and it does not, in general, correspond to a dimension of any subspace (the correspondence remains only if all of the coarse-grainings commute, in which case a joint subspace exists).

Observational entropy satisfies two important properties~\cite{safranek2019long},
\begin{align} 
S_{\mathrm{vN}}\leq S_{\C_1,\dots,\C_n}\leq \ln \mathrm{dim}\HS,\label{eq:bounds}\\ 
S_{\C_1,\dots,\C_n,\C_{n+1}}\leq  S_{\C_1,\dots,\C_{n}}. \label{eq:nonincreasing} 
\end{align}
The first property shows that observational entropy $S_{\C_1,\dots,\C_n}$ is upper bounded by the maximal uncertainty allowed by the size of the system, and lower bounded by the uncertainty inherent to the system (measured by von Neumann entropy $S_{\mathrm{vN}}=-\tr[\R\ln\R]$). The second property shows that every additional measurement can only decrease the entropy. These properties justify interpreting observational entropy as a measure of uncertainty an observer making measurements associates with a system. In more detail, observational entropy is the uncertainty an observer would associate to the initial state of the system if they had infinitely many copies of this state, and performed sequential measurements in bases $\C_1,\dots,\C_n$ on each copy. This would allow them to build up statistics of measurement outcomes, and thus determine this entropy exactly.

\begin{table*}
  \caption{Different types of measurements and their corresponding coarse-grainings.}\label{tab:typesofmeasurements}
  \begin{tabularx}{\textwidth}{llll}
  \hline
    Measurement  & Operator & Superoperator form \\
     & form &  \\
     \hline
    Projective  & $\C=\{\P_i\}$ & $\C=\{\P_i(\bullet)\P_i\}$ \\
    Kraus Rank-1 Generalized  & $\C=\{\K_i\}$ & $\C=\{\K_i(\bullet)\K_i^\dag\}$ \\
    Generalized  & None  & $\C=\{\PO_i\}$,\quad\ $\PO_i=\sum_m\K_{im}(\bullet) \K_{im}^\dag$ \\
    Multiple Generalized & None  & $\bC=(\C_1,\dots,\C_n)$,\ \ $\C_k=\{\PO_{i_k}\}$,\ \ \ $\PO_{i_k}=\sum_{m_k}\K_{i_km_k}(\bullet) \K_{i_km_k}^\dag$\\
    &  & This is equivalent to a single coarse-graining with \\
    &  & vector-labeled elements, $\bC=\{\PO_\bi\}$, $\PO_\bi=\PO_{i_n}\cdots\PO_{i_2}\PO_{i_1}$.\\
    %&  &
    %$\PO_\bi=\PO_{i_n}\cdots\PO_{i_2}\PO_{i_1}=\sum_{\bm}\K_{\bi\bm}(\bullet) \K_{\bi\bm}^\dag$\\
    %    &  &
    %and  $\K_{\bi\bm}=\K_{i_nm_n}\cdots\K_{i_1m_1}$.\\
  \end{tabularx}
\end{table*}

\subsection{Observational entropy with generalized measurements}
The definition and the related properties above include only projective measurements which are not the most general measurement an observer can perform. A generalized measurement includes the case of a projective measurement, but it also includes situations in which the system is allowed to interact with an auxiliary system (probe), and then a joint projective measurement is performed on the system plus the probe. Further, it includes cases where indirect measurements are performed only on the probe. Such measurements are described by a set of linear superoperators (\emph{quantum operations}/\emph{quantum instruments}) $\{\PO_i\}$. Each element of this set can be expressed in terms of its Kraus decomposition\footnote{Operators $\hPi_i=\sum_m\K_{im}^\dag\K_{im}$ are called POVM elements. Their collection $\{\hPi_i\}$ a POVM (\emph{positive operator-valued measure}) in the literature. Unlike the general measurement given by $\{\PO_i\}$, POVM by itself cannot determine the post-measurement state~\cite{kraus1983states}. 
},
\[\label{eq:kraus_dec}
\PO_i(\R)=\sum_{m}\K_{im}\R\K_{im}^\dag,
\]
%where $N_i$ is called the Kraus rank, 
in which the Kraus operators satisfy the completeness relation\footnote{Number of $m$ for each $i$ in decomposition~\eqref{eq:kraus_dec} is not unique, and the number of $m$ for each $i$ can vary. Minimal number of $m$ such that the decomposition holds is called the \emph{Kraus rank} of $\PO_i$. Projective measurements are a special type of Kraus rank-1 measurements defined by $\PO_i(\hat{X})\equiv \P_i\hat{X}\P_i$.} ,
\[\label{eq:completeness}
\sum_{i}\sum_{m}\K_{im}^\dag\K_{im}=\I.
\]
Upon obtaining a measurement outcome ``$i$'', the density matrix of the system is projected onto
\[
\R\xrightarrow{\text{``$i$''}}\frac{\PO_i(\R)}{p_i}
\]
with probability $p_i=\tr[\PO_i(\R)]$. 

In order to define observational entropy to include generalized measurements, we define each coarse-graining as a set of quantum operations $\C_k=\{\PO_{i_k}\}$. % satisfying the completeness relation~\eqref{eq:completeness}. 
Correspondingly, observational entropy is defined as,
\[
\label{eq:def3}
S_\bC\equiv S_{\C_1,\dots,\C_n}\equiv-\sum_{\bi}p_\bi\ln \frac{p_\bi}{V_\bi},
\]
where $\bC=(\C_1,\dots,\C_n)$ is a vector of coarse-grainings. The probability $p_\bi$ of obtaining a sequence of outcomes $\bi$ is obtained by combining the superoperators\footnote{For better visibility, we dropped excessive parentheses, and superoperator on the left is understood as acting on the superoperator on the right, such as in $\PO_{i_2}\PO_{i_1}(\R)\equiv \PO_{i_2}(\PO_{i_1}(\R))$.} ${\PO_\bi(\bullet)\equiv\PO_{i_n}\dots\PO_{i_2}\PO_{i_1}}(\bullet)$
such that
\[\label{eq:defp}
p_\bi=\tr[\PO_\bi(\R)].
\]
Further, we define
\[\label{eq:defV}
V_\bi=\tr[\PO_\bi(\I)]
\]
as the corresponding volume of a macrostate ($\I$ being the identity matrix). From Eqs.~\eqref{eq:defp} and~\eqref{eq:defV}, we view the vector of coarse-grainings as a single coarse-graining with vector-labeled elements $\bC\equiv(\C_1,\dots,\C_n)\equiv\{\PO_\bi\}$. The original definition is obtained by considering coarse-grainings made of projective superoperators\footnote{By $\C=\{\P_i(\bullet)\P_i\}$ we mean that $\C=\{\PO_i\}$, where $\PO_i(\hat{X})=\P_i\hat{X}\P_i$ for an operator $\hat{X}$.} $\C=\{\P_i(\bullet)\P_i\}$. Therefore, each projective coarse-graining is lifted into its superoperator form and we will use these two notations interchangeably ($\C=\{\P_i\}=\{\P_i(\bullet)\P_i\}$). See Table~\ref{tab:typesofmeasurements} for the overview of different types of measurements and their corresponding coarse-grainings. 

\subsection{Main results: Properties of observational entropy with generalized measurements}
What if, let us say, Eq.~\eqref{eq:bounds} or~\eqref{eq:nonincreasing} did not hold for the general case described above? For example, consider that there is a generalized measurement, such as an indirect measurement using a probe, that can push observational entropy below the inherent uncertainty of the system, violating Eq.~\eqref{eq:bounds}. Or performing a measurement results in an increase in the observers' uncertainty (making them ``forget'' something about the system), violating Eq.~\eqref{eq:nonincreasing}? Such outcomes would imply that observational entropy is not a good measure of the observers' uncertainty. However, in this section, we show that even with the inclusion of generalized measurements, the two properties [Eqs.~\eqref{eq:bounds} and \eqref{eq:nonincreasing}] still hold.

To state the theorems that are proved in Append.~\ref{app:proofs}, we first define:
\begin{definition}\label{def:c_g_by_observable} (Coarse-graining defined by an observable -- a Hermitian operator): Assuming spectral decomposition of a Hermitian operator $\hat{A}=\sum_a a\P_{a}$ (where $a$' are assumed to be distinct), we define coarse-graining given by the Hermitian operator as $\C_{\hat{A}}=\{\P_a(\bullet)\P_a\}$.
\end{definition}

Second, we define POVM elements given a coarse-graining, i.e.,
\begin{definition}\label{def:POVM_elements} For a vector of coarse-graining $\bC\equiv(\C_1,\dots,\C_n)=\{\PO_\bi\}$, where
$\PO_{\bi}=\sum_{\bm}\K_{\bi\bm}(\bullet)\K_{\bi\bm}^\dag$ and $\K_{\bi\bm}=\K_{i_nm_n}\cdots\K_{i_1m_1}$, we define POVM elements $\{\hPi_\bi\}$ as
\[
\hPi_\bi=\sum_\bm\K_{\bi\bm}^\dag\K_{\bi\bm}.
\]
This allows us to write the probabilities of outcomes and related volumes as $p_\bi=\tr[\hPi_\bi\R]$ and $V_\bi=\tr[\hPi_\bi]$.
\end{definition}

Third, we define the notion of finer coarse-graining. Intuitively, the finer coarse-graining is such that it provides at least as much information as a coarse-graining, irrespective of the specific state of the system. The following Theorem will show that we can also view the finer coarse-graining as that which never increases the observational entropy. The definition goes as,
\begin{definition}\label{def:finer_set_coarse_graining}~(Finer vector of coarse-grainings)\footnote{Alternative version: ${\bC}\hookleftarrow {\widetilde{\bC}}$ iff there exists a function $\mathcal{J}:\bi\rightarrow \bj$ such that for all $\bj$, $\hPi_\bj=\sum_\bi \delta_{\bj \mathcal{J}(\bi)}\hPi_\bi$, where $\delta_{\bj \tilde{\bj}}$ denotes the Kronecker delta.}: We say that a vector of coarse-grainings $\bC=\{\PO_\bi\}$ is finer than a vector of coarse-grainings $\widetilde{\bC}=\{\PO_\bj\}$ (and denote ${\bC}\hookleftarrow {\widetilde{\bC}}$ or ${\widetilde{\bC}}\hookrightarrow {\bC}$), when the POVM elements of $\widetilde{\bC}$ can be built from the POVM elements of $\bC$, i.e., when we can write
\[
\hPi_\bj=\sum_{\bi\in I^{(\bj)}}\hPi_\bi,\quad \quad \quad \quad (\forall \bj),
\]
where $I^{(\bj)}$ are disjoint index sets whose union is the set of all indices $\bigcup_\bj I^{(\bj)}=\{\bi\}$.
\end{definition}

It is clear from this definition that for ${\bC}\hookleftarrow {\widetilde{\bC}}$, the probabilities of outcomes and related volumes can be also written as sums, $p_\bj=\sum_{\bi\in I^{(\bj)}}p_\bi$ and $V_\bj=\sum_{\bi\in I^{(\bj)}}V_\bi$. This creates an impression that the finer coarse-graining offers the lens by which the state of the system is studied. This is because the Hilbert space is cut into smaller volumes with a finer coarse-graining, and each of this volume has its own associated probability $p_\bi$. On the other hand, the coarser coarse-graining just adds up together to create $p_\bj$, ignoring this finer structure. With this definition, it is straightforward to realize that there are some coarse-grainings that are equivalent in the sense that both ${\bC}\hookleftarrow {\widetilde{\bC}}$ and ${\bC}\hookrightarrow {\widetilde{\bC}}$. For example, such cases are $\bC=\{\PO_\bi\}$ and ${\widetilde{\bC}}=\{\widetilde\PO_\bi\}$, where $\PO_\bi=\sum_\bm\K_{\bi\bm}(\bullet)\K_{\bi\bm}^\dag$ and $\widetilde\PO_\bi=\sum_\bm \U_{\bi\bm}\K_{\bi\bm}(\bullet)\K_{\bi\bm}^\dag \U_{\bi\bm}^\dag$. The only difference between these two coarse-grainings is that with the second coarse-graining, the state was unitarily evolved with unitary operators $\U_{\bi\bm}$ after the measurement. These two coarse-grainings have the same POVM elements and thus are equivalent.

This definition of finer coarse-graining, while simple-looking, represents a major stepping stone in our understanding. It looks different from the previously introduced Definitions 2 and 6 in Ref.~\cite{safranek2019long}, which considered vectors of only projective coarse-grainings. However, as we show in Append.~\ref{app:equivalence_definitions}, it turns out that this definition is equivalent to the older definition when applied on the same restricted set, and therefore represents its direct generalization.

Generalizations of Theorems 7 and 8 from Ref.~\cite{safranek2019long} [represented here by Eqs.~\eqref{eq:bounds} and~\eqref{eq:nonincreasing}] follow.

\begin{theorem}\label{thm:bounded_multiple} Observational entropy~\eqref{eq:def3} with multiple coarse-grainings is bounded, i.e.,
\[
S_{\mathrm{vN}}\leq S_\bC\leq \ln \mathrm{dim}\HS,
\]
for any vector of coarse-grainings $\bC=(\C_1,\dots,\C_n)$ and any density matrix $\R$. $S_{\mathrm{vN}}= S_\bC$ if and only if $\bC\hookleftarrow \C_{\R}$. $S_\bC=\ln \mathrm{dim}\HS$ if and only if $\forall \bi$, $p_{\bi}=V_{\bi}/\mathrm{dim}\HS$.
\end{theorem}
The meaning of the first equality condition in Theorem~\ref{thm:bounded_multiple}, $\bC\hookleftarrow \C_{\R}$, can be inferred from Definition~\ref{def:c_g_by_observable} applied to Hermitian operator $\R$. It means that for every eigenvalue $\rho$ of the density matrix $\R=\sum_\rho\rho\P_\rho$, the corresponding projector onto an eigenspace can be written using the POVM elements as \[\label{eq:rhoandPi}
\P_\rho=\sum_{\bi\in I^{(\rho)}}\hPi_\bi.
\]
This means that an informationally complete measurement is any measurement that is finer than the measurement in the eigenbasis of the density matrix. Moreover, the above identity represents a connection between observational entropy and state identification, since it can be used to infer the quantum state in case when one can find a coarse-graining $\C$ for which $S_{\mathrm{vN}}= S_\bC$. This connection is described in detail in Sec.~\ref{sec:quantum_tomography}.

\begin{theorem}\label{thm:non-increase} Observational entropy~\eqref{eq:def3} is non-increasing with each added coarse-graining, i.e.,
\[
S_{\bC,\C_{n+1}}\leq S_{\bC},
\]
for any vector of coarse-grainings $(\bC,\C_{n+1})$ and any density matrix $\R$. The inequality becomes an equality if and only if $\forall \bi,i_{n+1}$, $p_{\bi,i_{n+1}}=(V_{\bi,i_{n+1}}/V_{\bi})p_{\bi}$.
\end{theorem}

The equality condition in Theorem~\ref{thm:non-increase} is satisfied, among other cases, when the sequence of measurements $\bC$ projects onto a pure state (meaning that all information about the initial state was depleted by the first $n$ measurements and there is no information left in the remaining state). Another case when the equality condition is satisfied is when $\bC^\dag\hookleftarrow \C_{n+1}^{\dag}$, expressing that the information the $(n+1)$th measurement could provide was already provided by the first $n$ measurements, making the $(n+1)$th measurement redundant. Above, we have defined $\bC^\dag\equiv(\C_n^\dag,\dots,\C_1^\dag)=\{\PO_\bi^\dag\}$, 
where $\PO_{\bi}^\dag\equiv\sum_{\bm}\K_{\bi\bm}^\dag(\bullet)\K_{\bi\bm}$
(see the end of the proof in Sec.~\ref{sec:proof2}).

Finally, we generalize an elegant and intuitive Theorem~2 from Ref.~\cite{safranek2019long}, which we also foreshadowed here when motivating the definition of finer coarse-graining:

\begin{theorem}\label{thm:monotonic} Observational entropy~\eqref{eq:def3} is a monotonic function of the coarse-graining.
If ${\bC}\hookleftarrow {\widetilde{\bC}}$ then
\[
S_{\bC}\leq S_{\widetilde{\bC}}.
\]
The inequality becomes an equality if and only if $\forall \bj,\forall \bi\in I^{(\bj)},p_{\bi}=(V_\bi /V_{\bj})p_{\bj}$.
\end{theorem}

The validity of these three theorems means that considering observational entropy as a measure of an observers uncertainty about the initial state of the system is conceptually justified. We will further consolidate this justification in Sec.~\ref{sec:quantum_tomography}, in connection with quantum state identification.

While the framework for projective and generalized measurements seem very similar and the main theorems hold for both, there are several non-trivialities that one encounters while trying to generalize observational entropy. The particular choice for the definitions delineated above is justified a-posteriori. That is, by showing that equivalent theorems can be proven when considering these particular definitions. We elaborate on these difficulties and other key differences between observational entropy with projective and general coarse-grainings in Append.~\ref{app:nontrivialities}.

\section{Application of observational entropy framework to general indirect measurement schemes%/measuring with probe
}\label{sec:app_indirect_measurement}

In this section, we apply our framework to find an optimal indirect measurement scheme by computing the observational entropy of a measurement that is performed on an auxiliary system (probe) that interacts with our system of interest. Assuming perfect control of the interaction between the system and the probe, we find that in order to extract all the information from the system, the dimension of the probe must be one dimension larger than the rank of the density matrix of the system.
Specifically, it is enough to use a two-dimensional probe to extract full information about any pure state of the system. This is in perfect alignment with previous findings of Ref.~\cite{PELLONPAA2017}, which studied very similar questions from the perspective of minimal normal measurement models. 

Consider a scenario where the system (system $A$) is first coupled to an auxiliary system (system $B$). As a consequence, the state of the auxiliary system (probe) is affected by the system. Alternatively, we could say that the probe collects some information about the system. Thus, measuring the probe provides information about the state of the system. This protocol can be schematically written as,
\[\label{eq:protocol}
\begin{split}
&\R\xrightarrow[]{}\R\otimes\hat{\sigma}\xrightarrow[]{}\U\R\otimes\hat{\sigma}\U^\dag\\
&\xrightarrow[]{\text{``$i$''}}\frac{(\I\!\otimes\!\P_i)\U\R\otimes\hat{\sigma}\U^\dag(\I\!\otimes\!\P_i)}{p_i}\xrightarrow[]{\mathrm{tracing\,out\,the\,probe}}\\
&\frac{\tr_B[(\I\!\otimes\!\P_i)\U\R\otimes\hat{\sigma}\U^\dag(\I\!\otimes\!\P_i)]}{p_i}\equiv\frac{\PO_i(\R)}{p_i}.
\end{split}
\]
Above, $\U$ is the unitary evolution operator that incorporates the interaction between the system and the probe, $\R$ is the initial density matrix of the system, and $\hat{\sigma}$ is that of the probe. Since only the probe is measured in this scenario, the projection operator $\P_i$ acts only on the probe, with the identity $\I$ being acted upon the system. The full Hilbert space is the tensor product of the system and the probe $\HS=\HS_A\otimes\HS_B$.

By tracing out the probe, this protocol gives an explicit form of the measurement superoperators acting on the system,
\[
\PO_i(\R)=\tr_A[(\I\otimes\P_i)\U\R\otimes\S \U^\dag(\I\otimes\P_i)].
\]
We denote coarse-graining consisting of these superoperators as 
\[
\Cvn={\{\PO_i\}}.
\]
This refers to coarse-graining given by indirect measurement scheme protocol, Eq.~\eqref{eq:protocol}, on the state of the system. In other words, this quantum operation acts solely on the density matrix of the system, and not on the probe.

By inserting the spectral decomposition of the probe $\S=\sum_{m}\sigma_m\pro{m}{m}$ we can also obtain the Kraus decomposition of the superoperator
\[\label{eq:superkraus}
\PO_i(\R)=\sum_{m,m'}\K_{imm'}\R\K_{imm'}^\dag,
\]
where $\K_{imm'}=\sqrt{\sigma_m}\bra{m'}(\I\otimes\P_i)\U\ket{m}$ are the Kraus operators.

While the Kraus operators can be useful in expressing the superoperator (and we will use this decomposition in our examples), we do not need this decomposition to calculate the observational entropy $S_{\Cvn}$. This is because $p_i$ and $V_i$ in Eq.~\eqref{eq:def3} depend only on the superoperator, thus we can obtain directly
\begin{align}
p_i&=\tr[\PO_i(\R)]=\tr[(\I\otimes\P_i)\U\R\otimes\hat{\sigma}\U^\dag(\I\otimes\P_i)],\nonumber\\
V_i&=\tr[\PO_i(\I)]=\tr[(\I\otimes\P_i)\U\I\otimes\hat{\sigma}\U^\dag(\I\otimes\P_i)].\label{eq:probabilityindirect}
\end{align}
Note that this $p_i$ is identical to the $p_i$ in the protocol expressed by Eq.~\eqref{eq:protocol}, and it represents the probability with which the outcome ``$i$'' was measured. Also, $V_i$ represents the corresponding volume in the Hilbert space of the system (\emph{not} the probe Hilbert space).

\subsection{Limitations of the finite-dimensional probe}\label{sec:optimal_indirect_measurement}

Now that we set up the general framework, we can ask the first natural question. Considering an indirect measurement, how low can the observational entropy get? Can it, for example, go as low as the ultimate minimum given by Theorem~\ref{thm:bounded_multiple}, that is, the von Neumann entropy, and therefore be as informative as a direct measurement that achieves this lower bound?

That of course depends on the amount of control over the system, as well as the size of the probe that is being used. In the case of perfect control, i.e., the ability to design any interaction unitary $\U$, the answer is fairly straightforward.

Note that the following protocols depend on the potentially unknown state of the system: a method that achieves these bounds would have to be applied adaptively, so these bounds could be reached in the limit of many copies of the initial state of the system. See more discussion on this in Sec.~\ref{sec:quantum_tomography}.

\subsubsection{High-dimensional probe}

First, we assume that the size of the probe is (at least) as big as the size of the system, i.e., $N\equiv \dim \HS_A=\dim \HS_B\equiv M$. We take the interaction to be a swap gate $\U(\R\otimes\S)=\S\otimes \R$, which defines $\U=\sum_{k,m=1}^{N}\pro{m,k}{k,m}$. This gives
\begin{align}
p_i&=\tr[(\I\otimes\P_i)\S\otimes\R(\I\otimes\P_i)]=\tr[\P_i\R],\\
V_i&=\tr[(\I\otimes\P_i)\S\otimes\I(\I\otimes\P_i)]=\tr[\P_i].
\end{align}
The corresponding observational entropy reduces to
\[
S_{\Cvn}=S_{\Cp},
\]
where $\Cp=\{\P_i\}=\{\P_i(\bullet)\P_i\}$ (in the operator and super-operator notation) denotes the projective coarse-graining on the probe. According to Theorem~\ref{thm:bounded_multiple}, taking $\Cp$ to be consisting of the projectors onto eigenvectors of the density matrix, $\Cp\hookleftarrow \C_\R$, leads to $S_{\Cp}=S_{\mathrm{vN}}$ and thus also $S_{\Cvn}=S_{\mathrm{vN}}$. Therefore, if we have a probe that is at least as big as the system, we can extract all the information by applying the swap gate and then measuring in the eigenbasis of the system density matrix. In other words, an indirect measurement is as good as a direct measurement in this case.

\subsubsection{Low-dimensional probe}

In the case when the probe is smaller than the system, $\dim \HS_B \equiv M\leq N \equiv \dim \HS_A$, we can obtain low entropy by diagonalizing the system and the probe density matrices and then performing a partial swap. We define $\U=\U_{\mathrm{partial\ swap}}(\U_{\mathrm{diag}\R}\otimes \U_{\mathrm{diag}\S})$, where
\begin{align}
\U_{\mathrm{diag}\R}\,\R\, \U_{\mathrm{diag}\R}^\dag&=\mathrm{diag}(\rho_1,\rho_2,\dots),\\ \U_{\mathrm{diag}\S}\,\S\, \U_{\mathrm{diag}\S}^\dag&=\mathrm{diag}(\sigma_1,\sigma_2,\dots),
\end{align}
with the eigenvalues being ordered as $\rho_1 (\sigma_1)\geq\rho_2(\sigma_2)\geq\dots$, and where
\[
\U_{\mathrm{partial\ swap}}\ket{k,m}=
\begin{cases}
\ket{m,k}& k\leq M,\\
\ket{k,m}& k>M.
\end{cases}
\]
In terms of eigenvectors this gives
\[
\U=\sum_{k,m=1}^{M}\pro{m,k}{k,m}+\sum_{k=M+1}^{N}\sum_{m=1}^M\pro{k,m}{k,m}.
\]
Above, $\R=\sum_{k=1}^N\rho_k\pro{k}{k}$ and $\S=\sum_{m=1}^M\sigma_m\pro{m}{m}$. This yields
\begin{align}
p_i&=\tr[\P_i(\R^{(\sigma)}+p\S)],\label{eq:pipartialswap}\\
V_i&=\tr[\P_i(\I+(N-M)\S)],\label{eq:Vipartialswap}
\end{align}
with $\R^{(\sigma)}=\sum_{m=1}^M \rho_m\pro{m}{m}$ being the part of the system density matrix that got swapped into the probe Hilbert space. The part that remains in the system Hilbert space is $\R^{(\mathrm{rest})}=\sum_{k=M+1}^N\rho_k\pro{k}{k}$, and $p=\tr[\R^{(\mathrm{rest})}]=\sum_{m=M+1}^N\rho_m$ is the sum of eigenvalues that remained in the system Hilbert space. Note that both $\R^{(\sigma)}+p\S$ and $\I+(N-M)\S$ are diagonal in the same basis $\{\ket{m}\}$.

Now that we specified our interaction, we need to investigate two additional factors over which we will optimize to obtain the lowest possible observational entropy: i) the state of the probe $\S$ and ii) the choice of measurement performed on the probe $\C_B=\{\P_i\}$.

First, we choose a probe that is in a completely mixed state $\S=\I/M$. In this case, the observational entropy is minimized when the coarse-graining on the probe Hilbert space is in the diagonal basis $\C_B=\{\P_i\}=\{\pro{m}{m}\}$\footnote{This corresponds to the system coarse-graining $\C=\{\tr_B[(\I\otimes\pro{m}{m}) \U(\bullet\otimes\S) \U^\dag (\I\otimes\pro{m}{m})]\}$.}. Using Eqs.~\eqref{eq:pipartialswap} and~\eqref{eq:Vipartialswap} we obtain
\[
S_{\Cvn}=-\sum_{i=1}^M\Big(\rho_i+\frac{p}{M}\Big)\ln\Big(\rho_i+\frac{p}{M}\Big)+\ln\frac{N}{M}.
\]
Depending on the size of the probe, this may be fairly close to the von Neumann entropy, and it will converge to it for $M\rightarrow N$ (since in this limit $p\rightarrow 0$).

Second, we consider a case when the probe is at least one dimension larger than the rank $R$ of the system density matrix, $M\geq R+1$. Here we we have $p=0$ (due to ordering $\rho_1\geq \rho_2\geq\dots$). We choose $\S=\pro{R+1}{R+1}$ and coarse-graining of the probe in the diagonal basis $\C_B=\{\P_i\}=\{\pro{m}{m}\}$ again. This yields
\[
p_i=\begin{cases}
\rho_i, &i\leq R,\\
0, &i> R,
\end{cases}\quad 
V_i=\begin{cases}
1, &i\!\neq\! R\!+\!1,\\
1\!+\!N\!-\!M, &i\!=\!R\!+\!1.
\end{cases}
\]
This results in
\[\label{eq:equalsvN}
S_{\Cvn}=-\sum_{i=1}^R\rho_i\ln \rho_i=S_{\mathrm{vN}}.
\]
Thus observational entropy achieves its minimum in this case.

\subsubsection{Takeaways from the analysis of a low-dimensional probe}
Equation~\eqref{eq:equalsvN} shows that when i) interaction $\U$ is chosen to be a partial swap gate, ii) the probe dimension is chosen as $R+1$, where $R$ is the rank of the system density matrix, and iii) the projective measurement performed on the probe is given by the eigenbasis of the system density matrix, then observational entropy equals the von Neumann entropy. We can intuitively rephrase this conclusion as follows:

\emph{The minimal uncertainty can be achieved by measuring a probe that is one dimension larger than the rank of the state of the system.}

%The fact that minimal uncertainty is achieved does not necessarily mean that by making that measurement we have determined the state of the system. To do that, we would need to additionally know the value of the von Neumann entropy. See more details in Sec.~\ref{sec:quantum_tomography}. We summarize this as follows:

%\emph{Knowing the value of von Neumann entropy, the state of the system can be exactly determined by measuring a probe which is one dimension larger than the rank of the state.}

As a special case, we conclude that any pure state of the system can be exactly determined by using a two-level probe.

Physically, this is fairly straightforward to comprehend. Making a projective measurement on a two-level probe leads to (at most) two possible outcomes. In an ideal case, when observational entropy is zero, we obtain the first outcome with $100\%$ probability, and we never see the second outcome. This is because the interaction unitary $\U$ translates the measurement on the probe into a measurement on the system as if one were to measure in the $\{\pro{\psi}{\psi},\I-\pro{\psi}{\psi}\}$ basis on the system, where $\R=\pro{\psi}{\psi}$ is the state of the system. In other words, eigenvectors of the density matrix with non-zero eigenvalues are mapped onto one state of the probe, and all the others onto the other state. %In other words, the wave function of the system is exactly mapped onto the first level of the probe, while the rest is mapped onto the second level of the probe. This clearly generalizes into states of higher rank, $\R=\sum_{m=1}^R\rho_m\pro{m}{m}$. There, the measurement on probe is the same as if one were to measure the system in the $\{\pro{1}{1},\dots,\pro{R}{R},\I-\sum_{m=1}^R\pro{m}{m}\}$ basis\footnote{The actual measurement is more complicated than this. To be exact, the von Neumann measurement scheme leads to
% $\PO_i(\R)=\K_i\R\K_i^\dag$, where
% $\K_i=\pro{R+1}{i}+\delta_{i,R+1}\sum_{k=M+1}^N\pro{k}{k}$ was computed from Eq.~\eqref{eq:superkraus}.  We denote this coarse-graining $\C_1=\{\K_i\}$. This is equivalent to measuring in basis $\C_2=\{\pro{1}{1},\dots,\pro{R}{R},\pro{R+2}{R+2},\dots,\pro{M}{M},\I-\sum_{m=1}^R\pro{m}{m}-\sum_{m=R+2}^M\pro{m}{m}\}$, because both measurement bases $\C_1$ and $\C_2$ lead to the same POVM elements: $\K_i^\dag\K_i=\pro{i}{i}$ for $i\neq R+1$ and $\K_{R+1}^\dag\K_{R+1}=\I-\sum_{m=1}^R\pro{m}{m}-\sum_{m=R+2}^M\pro{m}{m}$. Thus, also all probabilities $p_i=\tr[\PO_i(\R)]=\tr[\K_i^\dag\K_i\R]$ are the same, and both measurement carry the same informational content. This equivalence is also in the sense of Definition~\ref{def:finer_set_coarse_graining}: the two coarse-grainings $\C_1$ and $\C_2$ are finer than each other. However, note that in these two coarse-grainings, the final (projected) state is different.}.

The explicit algorithm that describes \emph{how} to infer the quantum state is detailed in Sec.~\ref{sec:quantum_tomography}.

Note that the findings of this section corroborate those of Ref.~\cite{PELLONPAA2017}, which also studied the minimal dimensionality of the auxiliary system for optimal measurements from the perspective of minimal normal measurement models.

\section{Application of the observational entropy framework to von Neumann measurement scheme}\label{sec:von_neumann}

In this section, we take an example of an indirect measurement scheme, called the von Neumann measurement scheme, and analyze it with the lens of the observational entropy framework. This is to illustrate the theory of observational entropy as a measure of observers' uncertainty and to demonstrate this measure as a useful figure of merit that helps in determining the type of measurement that performs best in information extraction. The derived expressions for the probabilities and corresponding volumes will be used for our simulations in the next section.

\subsection{Single measurement von Neumann scheme}

A special type of indirect measurement scheme, given by protocol of Eq.~\eqref{eq:protocol}, is the von Neumann measurement scheme (see Fig.~\ref{Fig:scheme} for illustration). In this scheme, the auxiliary system (\emph{probe}, also called \emph{pointer} specifically in this scheme), which is initially decoupled from the system, is assumed to be a heavy mass particle (although represented by a quantum state) initially localized at some position $x$ with some uncertainty $\Omega$. It interacts with the system through a unitary,
\[\label{eq:interactionU}
\U=\exp\big(-i\kappa\hat{M}\otimes \hat{p}\big).
\]
Above, $\hat{M}$ is an Hermitian system observable, and 
\[
\hat{p}=-i\frac{d}{dx}
\]
is the momentum operator of the probe~\footnote{The action of the unitary on the position-represented probe state is to translate it, i.e., $e^{-ix_0\hat{p}}\ket{\varphi_x}=\ket{\varphi_{x-x_0}}$.}. The parameter $\kappa=\lambda \Delta\tau$ is a product of interaction strength $\lambda$ and the time of interaction $\Delta\tau$. Using the eigenstates of the observable ($\hat{M}\ket{m} = \mu_m \ket{m}$), the operator $\hat{M}$ can be represented as $\hat{M}=\sum_\mu \mu \P_\mu$, where $\P_\mu=\sum_{\mu_m=\mu}\pro{m}{m}$ is a projector. Writing the state of the system in the eigenbasis of the measurement operator as $\ket{\psi}=\sum_m\alpha_m\ket{m}$, the interaction unitary acts as\footnote{Abusing the notation, what we mean by the present expressions is $\ket{\varphi_x}\equiv\int_{-\infty}^{\infty}dx\, \varphi_x\ket{x}$ and $\ket{\varphi_{x-\kappa\mu_ m}}\equiv\int_{-\infty}^{\infty}dx\, \varphi_{x-\kappa\mu m}\ket{x}$, when expressed in the position basis.}
\[\label{eq:r}
\U\ket{\psi}\ket{\varphi_x}=\sum_{m} \alpha_m\ket{m}\ket{\varphi_{x-\kappa \mu_m}}
\]
on a initial product state.

\begin{figure}[t]
\begin{center}
\includegraphics[width=1\hsize]{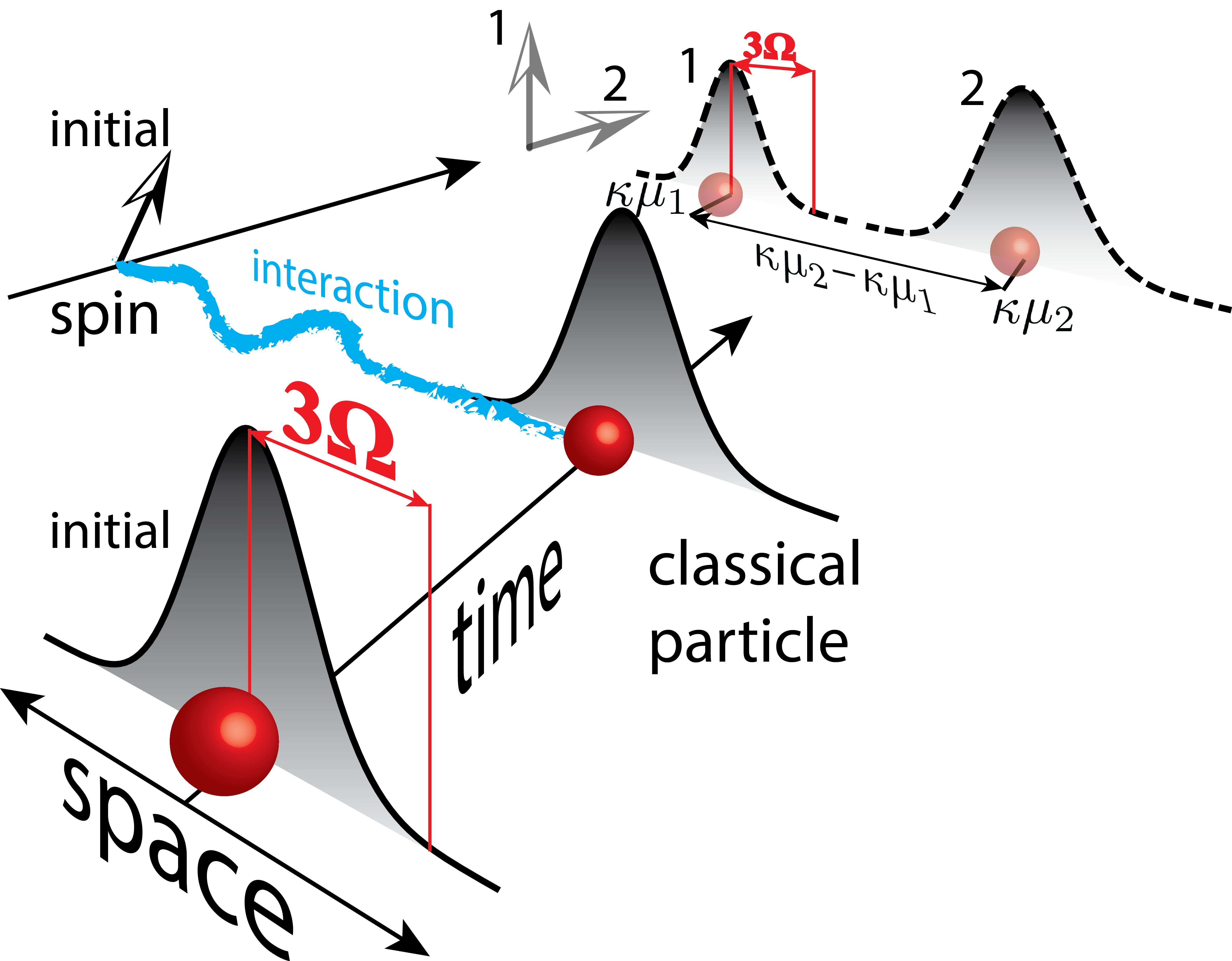}\\
\caption
{
Schematic picture of the von Neumann measurement scheme. A classical particle is localized at position $x=0$. Its position uncertainty is given by the Gaussian distribution with standard deviation $\Omega$. Through the interaction Eq.~\eqref{eq:interactionU}, the spin becomes correlated with the classical particle: if the spin is pointing up, the classical particle is moved by $\kappa \mu_1$, where $\mu_1$ is the first eigenvalue of the measurement operator $\hat{M}$ (in this case the spin operator). If the spin is pointed horizontally, the particle is moved by $\kappa \mu_2$ (with $\mu_2$ being the second eigenvalue). Thus, the spin information is translated into a position measurement of the classical particle.}
\label{Fig:scheme}
\end{center}
\end{figure}

For a general decoupled system and probe, we have
\[\label{eq:Uongeneralrho}
\begin{split}
&\U\R\otimes\pro{\varphi_x}{\varphi_x}\U^\dag\\
&=\sum_{m,m'}\rho_{mm'}\pro{m}{m'}\otimes \pro{\varphi_{x-\kappa \mu_m}}{\varphi_{x-\kappa \mu_{m'}}},    
\end{split}
\]
where $\R=\sum_{m,m'}\rho_{mm'}\pro{m}{m'}$ is the mixed initial state of the system.

After the interaction of the system and the probe, the position of probe is measured. This measurement corresponds to a coarse-graining, $
\C_B=\{\P_x\}$ $(\P_x=\pro{x}{x})$, on the probe. The von Neumann measurement scheme translates measuring of a system observable $\hat{M}$ into a problem of measuring the position of the probe. To illustrate that, consider an extreme case in which the probe is completely localized, i.e., its state is given by a position eigenstate $\ket{\varphi_x}=\ket{x}$. The outcome of a position measurement on the probe after the interaction is ``$x-\kappa \mu$'' with probability $p_\mu=\tr[\P_\mu\R]=\sum_{\mu_m=\mu}\rho_{mm}$. After obtaining this measurement outcome, the system state is projected onto the state $\P_\mu\R\P_\mu/p_\mu$. Thus, in such an idealized case in which the probe is completely localized, the von Neumann measurement scheme corresponds to a projective measurement of observable $\hat{M}$ on the system.

The protocol of the von Neumann measurement scheme can be summarized by a map
\[\label{eq:protocolvonNeumann}
\begin{split}
&\R\xrightarrow[]{\text{``$x$''}}\frac{\tr_B[(\I\!\otimes\!\P_x)\U\R\otimes\S \U^{\dag}(\I\!\otimes\!\P_x)]}{p_x}\equiv\frac{\PO_x(\R)}{p_x},
\end{split}
\]
where $\U=\exp(-i \kappa\hat{M}\otimes \hat{p})$, $p_x=\tr[\PO_x(\R)]$, and $\C=\{\PO_x\}$ defines the coarse-graining on the system. It is assumed that we can measure the probe exactly, so there is no uncertainty coming from measuring the position directly. However, typically the probe itself is assumed not to be fully localized (in order to imitate the quantum uncertainty in measuring the position); in this paper, we choose a probe initially prepared in a \emph{pure Gaussian} state,
\[\label{eq:gaussianancilla}
\S=\pro{\varphi_x}{\varphi_x},\quad \ket{\varphi_x}=\int dx \sqrt{g^\Omega_x}\ket{x},
\]
where 
\[
g^\Omega_x=\frac{1}{\sqrt{2\pi \Omega^2}}\exp\bigg(-\frac{x^2}{2\Omega^2}\bigg),
\]
and $\Omega>0$. A completely localized probe corresponds to $\Omega\rightarrow 0$. This probe state is a common choice in the literature~\cite{von1955mathematical,talkner2018measurement,thingna2020,son2021monitoring} due to its straightforward interpretation and possibility to realize such a state in an experiment~\cite{becerra2013implementation,roncaglia2014nonequilibrium}.

\subsection{Other types of von Neumann schemes}

The protocol that we described so far is called the \emph{single measurement} von Neumann scheme. There are other, more complicated types of von Neumann measurement schemes~\cite{thingna2020}. We specifically name \emph{repeated measurements}, in which multiple probes are used sequentially, and \emph{repeated contacts} scheme, in which the probe is reused, interacting with the system several times before being measured. We use observational entropy as a quantifier and illustrate the influence of different measurement schemes on information extraction from a quantum system.

To highlight the differences, let us summarize an observation from Ref.~\cite{thingna2020}: in situations with no free-time evolution in between measurements/contacts, both repeated measurements and repeated contacts can be used to extract the full available information. Repeated measurements have the advantage of collecting information about the system continuously (through many measurements) and therefore updating our information about the quantum system. However, due to the frequent measurements, the system dynamics are significantly affected by the back action for a large number of interactions, resulting in slower acquisition of information. 

In repeated contacts, we lose the ability to track the quantum system continuously, because we perform only one measurement in the end, however, the acquisition of information by the probe is faster, i.e., $N$ contacts provide more information than $N$ measurements. We will see this behavior exactly quantified by observational entropy in Sec.~\ref{sec:numerics}, as well as other interesting scenarios. For example, when the free time evolution does not commute with the measurement operator, there can be a back flow of information from the probe to the system. This information loss from the probe can be quantified using observational entropy. 

The different types of von Neumann measurement schemes considered in this work are illustrated in Fig.~\ref{Fig:types} and their mathematical descriptions are summarized in the next subsection and in Append.~\ref{app:explicit_forms}.

\begin{figure}[t]
\begin{center}
\includegraphics[width=1\hsize]{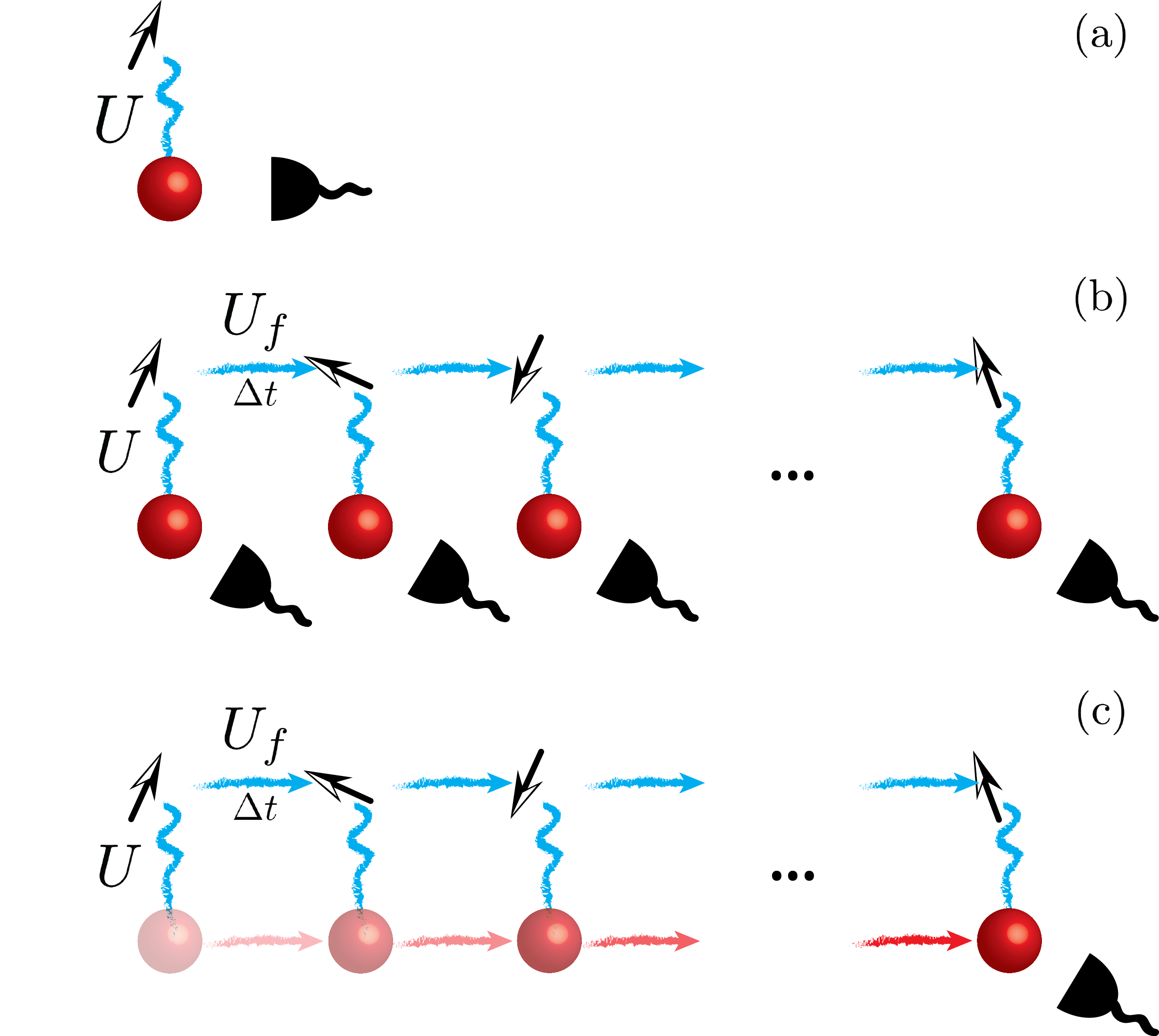}\\
\caption
{
Different types of von Neumann measurement schemes. (a) Single measurement: the quantum system (illustrated by an arrow denoting spin) interacts with a single probe (illustrated by a red circle) that is then measured, revealing information about the system. (b) Repeated measurements: the quantum system interacts with $N$ different probes. After each interaction, the probe that just interacted is measured. The system evolves with the free Hamiltonian for time $\Delta t$ in between the interactions. (c) Repeated contacts: the quantum system interacts $N$ times with the same single probe. This probe retains its state in between the interactions while the system evolves with the free Hamiltonian. At the very end, the probe is measured.}
\label{Fig:types}
\end{center}
\end{figure}

\subsection{Observational entropy in von Neumann measurement schemes}

\subsubsection{Projective measurement}

Before diving into observational entropy of von Neumann measurement schemes, we first define observational entropy of a projective measurement, which is performed directly on the system itself. The value of this entropy will serve as a reference to which we compare entropies of the different types of von Neumann measurement schemes. The probabilities of outcomes of a projective measurement and related volumes are given by
$p_\mu=\tr[\P_\mu\R]$ and $V_\mu=\tr[\P_\mu]$ respectively. This defines observational entropy of a projective measurement as $S^{\mathrm{pm}}\equiv S_{\C_{\hat{M}}}=-\sum_\mu p_\mu\ln(p_\mu/V_\mu)$.

\subsubsection{Single measurement}\label{sec:singlevN}

The evolution superoperator of a single von Neumann measurement will be derived by combining Eqs.~\eqref{eq:Uongeneralrho} and~\eqref{eq:gaussianancilla}, which gives
\begin{eqnarray}\label{eq:vonNeumannmap}
    \PO_x(\R)&=&\bra{x}\U\R\otimes\pro{\varphi_x}{\varphi_x}\U^\dag\ket{x} \nonumber \\
    &=&\sum_{m,m'}\rho_{mm'} \braket{x}{\varphi_{x-\kappa\mu_{m}}}\braket{\varphi_{x-\kappa\mu_{m'}}}{x} \pro{m}{m'} \nonumber \\
    &=&\sum_{m,m'}\rho_{mm'} \sqrt{g^\Omega_{x-\kappa\mu_{m}}}\sqrt{g^\Omega_{x-\kappa\mu_{m'}}} \pro{m}{m'}\nonumber \\
    &=&\K_x\R\K_x^\dag,
\end{eqnarray}
where $\K_x=\sum_m \sqrt{g^\Omega_{x-\kappa\mu_m}}\pro{m}{m}$. Clearly, from the normalization of Gaussian function we have $\int_{-\infty}^{+\infty}\K_x^\dag\K_x\,dx=\I$.
This further yields
\begin{align}
p_x&=\tr[\PO_x(\R)]=\sum_\mu\sum_{\mu_m=\mu}\rho_{mm}g^\Omega_{x-\kappa \mu}=\sum_\mu p_\mu g^\Omega_{x-\kappa\mu},\nonumber\\
V_x&=\tr[\PO_x(\I)]=\sum_\mu\sum_{\mu_m=\mu}g^\Omega_{x-\kappa \mu}=\sum_\mu V_\mu g^\Omega_{x-\kappa\mu}.\label{eq:Vx}
\end{align}
These probabilities and volumes will be used for calculating observational entropy of a single von Neumann measurement $S^{\mathrm{sm}}=-\int dx \ p_x\ln(p_x/V_x)$. Using Jensen's inequality, one can easily prove $S^{\mathrm{pm}}\leq S^{\mathrm{sm}} \leq \ln \dim \HS$, while the bounds $S^{\mathrm{pm}}$ and $\ln \dim \HS$ are reached by $S^{\mathrm{sm}}$ in limits $\Omega\rightarrow 0$ and $\Omega\rightarrow +\infty$, respectively. The results for this single measurement protocol are shown in Fig.~\ref{Fig:single} using a two-level system example whose details are described in Sec.~\ref{sec:genpar}.

\subsubsection{Repeated measurements}\label{sec:rep_measurements}
Repeated von Neumann measurement consists of $N$ interactions with $N$ different probes, all prepared in the same initial state~\eqref{eq:gaussianancilla}. The position of the probe is measured immediately after the system interacts with it. The time of interaction with each probe is $\Delta \tau=\kappa/\lambda$, chosen to be much smaller than the evolution time-scale of the system in between two interactions. When this process is repeated $N$ times with $N$ different probes, we obtain $N$ measurement outcomes of the positions of the probes, $\bx=(x_1,\dots,x_n)$. Assuming that the system freely evolves with Hamiltonian $\hat{H}$ in between the interactions, $\U_f=\exp(-i\hat{H}\Delta t)$ of the system, the total time of free evolution $T=N\Delta t$. Additionally, we assume that the probes themselves do not evolve outside of the interaction (the timescale of their evolution is much slower than the timescale of the experiment). The above description can be illustrated as follows,
\[\label{eq:protocolrepeatedinteractions}
\begin{split}
&\R\xrightarrow[]{\text{``$\bx$''}}\frac{\PO_\bx^{\mathrm{rm}}(\R)}{p_\bx},
\end{split}
\]
where, defining evolution superoperator $\mathcal{U}_f=\U_f(\bullet)\U_f^\dag$, the superoperator of repeated measurement is
\[\label{eq:risuper}
\PO_\bx^{\mathrm{rm}}(\bullet)=\mathcal{U}_f\PO_{x_N}\dots\mathcal{U}_f\PO_{x_1}(\bullet).
\]
Above, $\PO_{x_i}$ are defined by Eq.~\eqref{eq:vonNeumannmap}. The probabilities and volumes for the observational entropy are $p_\bx=\tr[\PO_\bx^{\mathrm{rm}}(\R)]$ and $V_\bx=\tr[\PO_\bx^{\mathrm{rm}}(\I)]$, which gives $S^{\mathrm{rm}}=-\int dx_1\dots dx_N \ p_\bx\ln(p_\bx/V_\bx)$.

\subsubsection{Repeated contacts}\label{sec:rep_contacts}
In sharp contrast to repeated measurements, repeated contacts consist of $N$ interactions with a single probe, \emph{without} resetting or measuring the probe after each interaction. The probe is measured only once with a position measurement at the very end. Similar to repeated measurements, here we also assume that the evolution time of the probe is very slow as compared to the timescale of the experiment and the interaction time interval $\Delta \tau = \kappa/\lambda$ is very short such that the system does not evolve within that time. It is described by 
\[
\begin{split} 
&\R\xrightarrow[]{\text{``$x$''}}\frac{\PO_x^{\mathrm{rc}}(\R)}{p_x}.
\end{split}
\]
The superoperator of repeated contacts reads
\[\label{eq:rcsuper}
\PO_x^{\mathrm{rc}}(\bullet)=\tr_B[(\I\otimes\P_x)(\mathcal{U}_f\mathcal{U})^N(\bullet\,\otimes\S) (\I\otimes\P_x)],
\]
where we have defined evolution superoperator $\mathcal{U}_f=(\U_f\otimes \I) (\bullet)(\U_f\otimes \I)^\dag$ and interaction superoperator $\mathcal{U}=\U (\bullet)\U^\dag$ with $\U$ being defined by Eq.~\eqref{eq:interactionU}.  Probabilities and volumes for the observational entropy are given by $p_x=\tr[\PO_x^{\mathrm{rc}}(\R)]$ and $V_x=\tr[\PO_x^{\mathrm{rc}}(\I)]$, which yields $S^{\mathrm{rc}}=-\int dx \ p_x\ln(p_x/V_x)$. 

There are several different scenarios of repeated contacts, depending on the time step $\Delta t$ in the evolution operator $\U_f$. In the case where the total time of interaction $T=N\Delta t$ and the ratio of interaction time and free evolution time $\kappa/\Delta t=\lambda \Delta \tau/\Delta t$ are fixed, we can solve this problem analytically for a large $N$. This solution is found to be
\[\label{eq:limit_PO}
\PO_x^{\mathrm{rc}}(\bullet)=\tr_B[(\I\!\otimes\!\P_x)\U_{\mathrm{limit}}(\bullet\,\otimes\S)\U_{\mathrm{limit}}^\dag (\I\!\otimes\!\P_x)].
\]
Above $\tr_B$ represents the trace over the probe and the unitary operator
\[
\U_{\mathrm{limit}}=\exp\big(-i\big(\hat{H}\otimes\I+R\,\hat{M}\otimes \hat{p}\big)T\big)
\]
was obtained from the Lie-Trotter product formula $e^{\hat{X}+\hat{Y}}=\lim_{N\rightarrow \infty}(e^{\hat{X}\frac{1}{N}}e^{\hat{Y}\frac{1}{N}})^N$. We denote the corresponding observational entropy as $S_{\mathrm{limit}}$. This particular value is plotted as a solid blue line in Fig.~\ref{Fig:limit}.

The explicit forms of the repeated measurement and repeated contact superoperators, which we used for our simulations in the next section, can be found in Append.~\ref{app:explicit_forms}.

\begin{figure}[t!]
\begin{center}
\includegraphics[width=1\hsize]{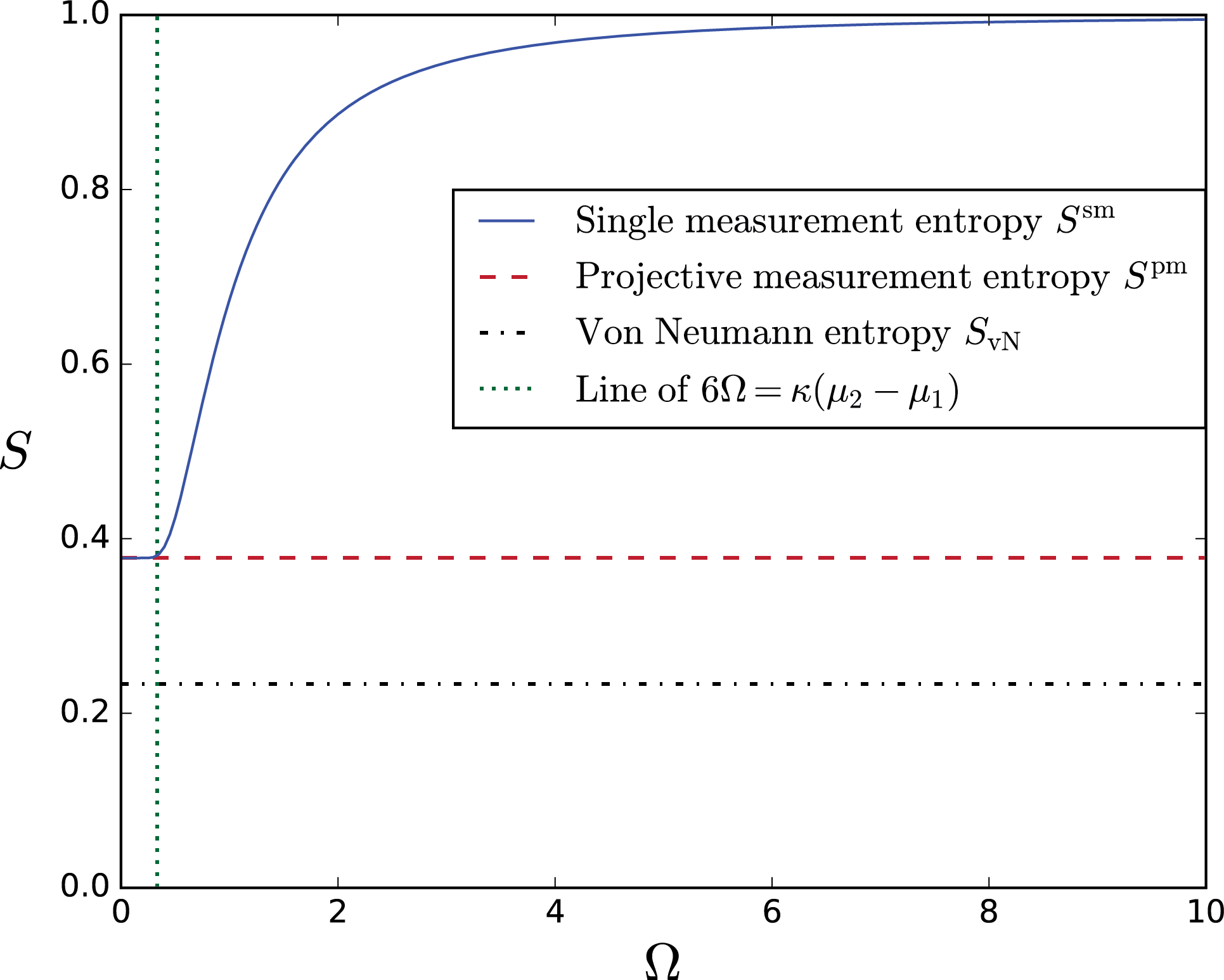}\\
\caption
{
Observational entropy of a single von Neumann measurement $S^{\mathrm{sm}}$ (blue curve) as a function of the standard deviation of the probe $\Omega$, compared to the observational entropy of related projective measurement $S^{\mathrm{pm}}$ (red dashed line) and the von Neumann entropy (black dot-dashed line). In the limit $\Omega\rightarrow 0$, $S^{\mathrm{sm}}=S^{\mathrm{pm}}$ [see discussion below Eq.~\eqref{eq:Vx}]. The green dotted vertical line marks the transition below which $S^{\mathrm{pm}} \approx S^{\mathrm{sm}}$. Within this regime, $6\Omega \le \kappa(\mu_2-\mu_1)$, the two Gaussian distributions of the probe, corresponding to the two measurement outcomes, barely overlap with each other (see Fig.~\ref{Fig:scheme}). This allows the observer to infer the quantum state of the system exactly. In other words, for highly (but not necessarily infinitely) localized probes, indirect measurement can extract as much information as a direct measurement. In the other limit, $S^{\mathrm{pm}}=1$ for $\Omega\rightarrow +\infty$ indicating that when the Gaussian probe is highly spread, inference about the system state is impossible, thus resulting in the maximal entropy. The initial state is chosen as $\R(0,\pi/16,\pi/16)$ [Eq.~\eqref{eq:initial_state}], $\kappa=1$, and the free evolution between contacts $\U_f = \I$.}
\label{Fig:single}
\end{center}
\end{figure}

\begin{figure*}[t]
\begin{center}
\includegraphics[width=1\hsize]{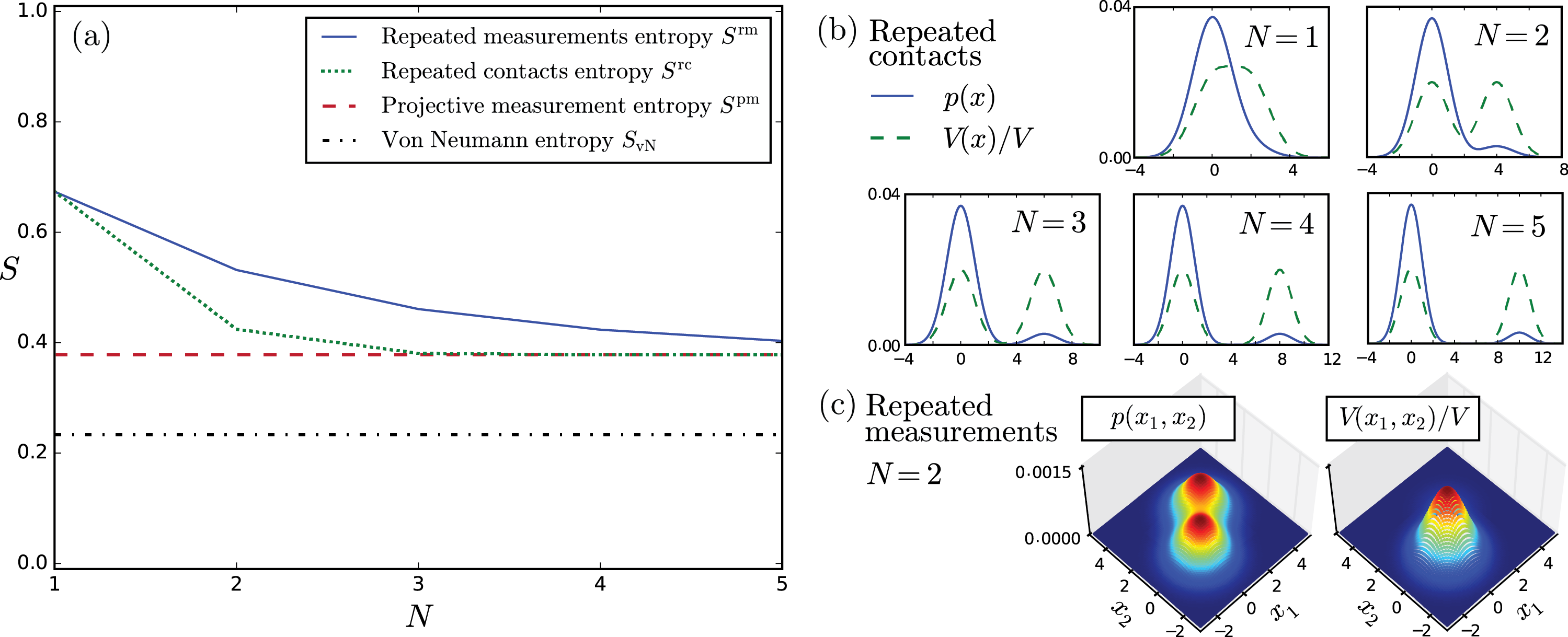}\\
\caption
{(a) The repeated measurement entropy $S^{\mathrm{rm}}$ (blue solid line) and the repeated contacts entropy $S^{\mathrm{rc}}$ (green dotted line) in comparison with the observational entropy given by the projective measurement $S^{\mathrm{pm}}$ (red dashed line), as a function of number of interactions or number of contacts $N$. The two-level system evolves with Hamiltonian~\eqref{eq:yesH} that commutes with the measurement operator $\hat{M}$. Both entropies converge to the projective measurement entropy for a large number of measurements/contacts. We choose the initial state to be $\R(0,\pi/16,\pi/16)$, Eq.~\eqref{eq:initial_state}, but the same qualitative behavior is observed for any initial state. For any initial state, for $N=1$ the corresponding observational entropies are equal, both converge to the projective measurement entropy, and the repeated contacts scheme always converges faster in $N$ than that of the repeated measurement scheme. The von Neumann entropy (black dash dotted line) is a lower bound to all other entropies, as expected from Theorem~\ref{thm:bounded_multiple}. (b) Corresponding probability distribution $p(x)$ for repeated contacts as a function of a position measurement outcome $x$, for various number of repeated contacts $N$, and the corresponding normalized macrostate volumes $V(x)/V$ (where $V=\dim \HS=2$). For increasing $N$ the distributions $p(x)$ and $V(x)/V$ become more distinguishable as expected from Eq.~\eqref{eq:KLdivergence}, which in this case specifically gives $S^{\mathrm{rc}}=\log_2 V-D_{\mathrm{KL}}(p(x)||V(x)/V)$. (c) The probability distribution of position measurements of two probes $p(x_1,x_2)$ and the corresponding weighted volumes [$N=1$ case is identical to repeated contacts as shown in (b). Outcomes $x_1$ and $x_2$ are correlated, but symmetric ($p(x_1,x_2)=p(x_2,x_1)$]. Also here we have $S^{\mathrm{rm}}=\log_2 V-D_{\mathrm{KL}}(p(\bx)||V(\bx)/V)$. In panel (a) we see that the repeated contacts entropy converges faster than the repeated measurements entropy. Intuitively, this is because it is easier to move the two Gaussians apart in one dimension to make them more distinguishable, as is shown in (b), rather than to increase the number of dimensions as shown in (c), in which case the overlap between $p(\bx)$ and $V(\bx)/V$ decays slower---expected to be roughly by a $\sqrt{N}$ slower (see the final paragraph of Sec.~\ref{sec:numerics}). For all panels we choose $\kappa=\Omega=1$, and the time step of free-time evolution between the measurements/contacts to be fixed at $\Delta t=1$. Thus the total time of free-time evolution is given by $T=N\Delta t$.}
\label{Fig:non}
\end{center}
\end{figure*}

\begin{figure*}[t!]
\begin{center}
\includegraphics[width=1\hsize]{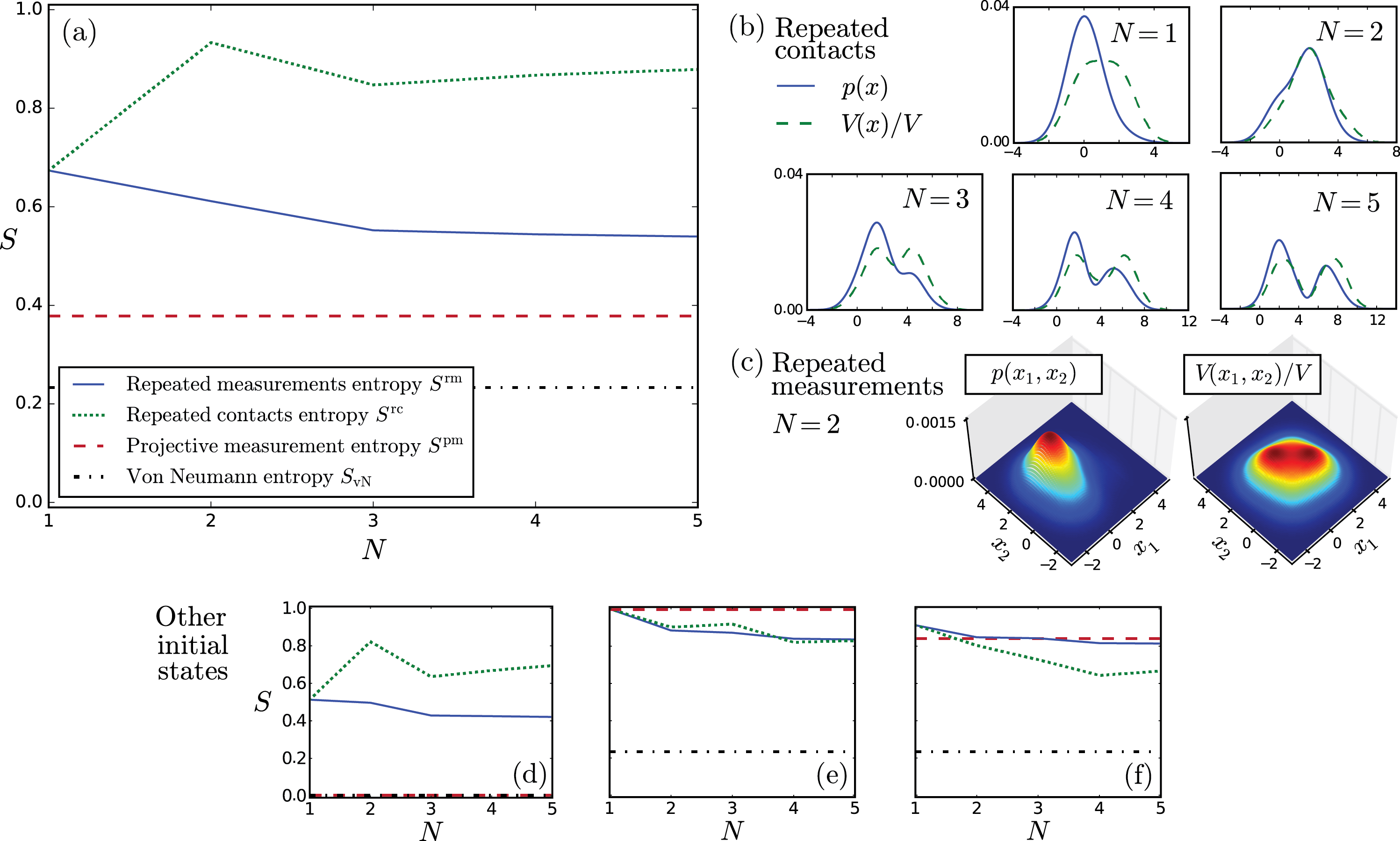}\\
\caption
{(a -- c) are similar to Fig.~\ref{Fig:non}(a -- c), respectively, with the same initial state $\R(0,\pi/16,\pi/16)$. However, here we choose Hamiltonian~\eqref{eq:nonH} that does not commute with the measurement operator $\hat{M}$. (d -- f) consider initial states $\R(0,0,0)$, $\R(0,\pi/4,\pi/16)$, $\R(\pi/3,\pi/3,\pi/16)$, Eq.~\eqref{eq:initial_state}, respectively, with other parameters identical to the case (a). (e) shows that repeated contact and repeated measurement entropies are not bounded by each other. In particular, unlike in the commuting case of Fig. 4, one cannot say that either repeated contacts or repeated measurement schemes perform better than the other (compare (a) and (d) versus (f)). Which performs better depends on the specific initial state, evolution, measurement operators, and the number of measurements/contacts (e). In (e) and (f), the evolution by a Hamiltonian that does not commute with the measurement operator pushes both the repeated measurement and repeated contact entropy below the projective measurement entropy. It does this by aligning the state better with the measurement basis through the evolution (making the eigenvectors of the transformed state more similar to the eigenvectors of the measurement operator $\hat{M}$ than it was before the evolution). All panels show that repeated measurement entropy always decreases with $N$ because each new measurement provides new information, and does not cause the observer to forget information they already obtained. This is not true for repeated contacts, because there is only a single measurement at the very end, so there may be times when the information leaks back from the probe to the system through the evolution, thus becoming inaccessible at the time of measurement. In all cases, entropies start from the same value, because for $N=1$ the repeated contacts scheme is equal to the repeated measurements scheme. (c) shows that unlike in the commuting case shown in Fig.~\ref{Fig:non}, in this non-commuting case, i.e., $[\hat{H},\hat{M}]\neq 0$, the probability distribution $p(x_1,x_2)$ is not symmetric. All other parameters are same as Fig.~\ref{Fig:non}.}
\label{Fig:int}
\end{center}
\end{figure*}

\begin{figure}[t!]
\begin{center}
\includegraphics[width=1\hsize]{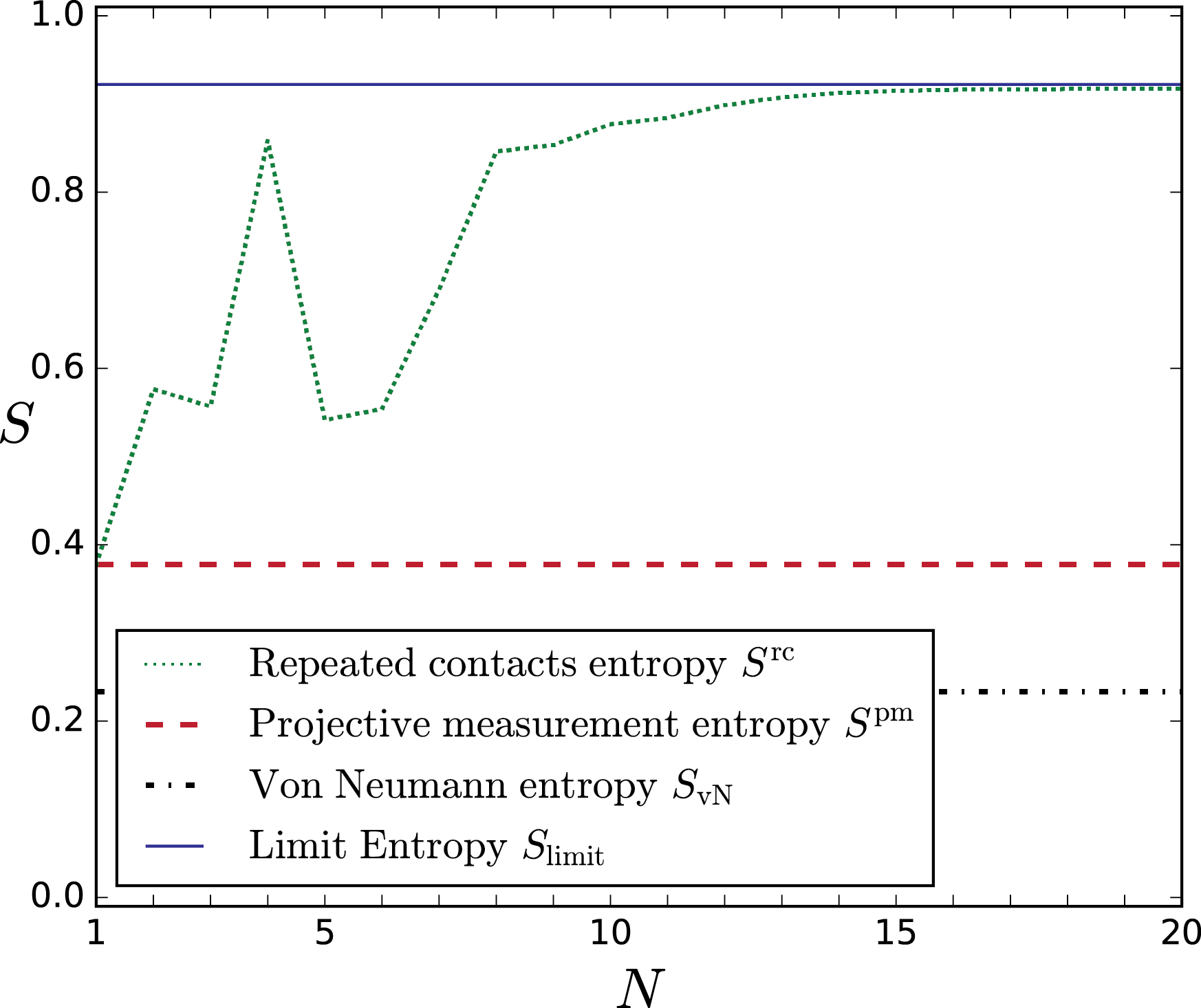}\\
\caption
{Behavior of repeated contact entropy $S^{\mathrm{rc}}$ as a function of the number of contacts $N$ where the total time of interaction $T=N\Delta t = 10$ is fixed (different limit as compared to Figs.~\ref{Fig:non} and~\ref{Fig:int} where $\Delta t$ was fixed). The ratio of the interaction and the free time evolution time is set as $\kappa/\Delta t=1$. The repeated contact entropy converges to the predicted theoretical value $S_{\mathrm{limit}}$ in the limit of large $N$, Eq.~\eqref{eq:limit_PO}.}
\label{Fig:limit}
\end{center}
\end{figure}

\section{Numerical results}\label{sec:numerics}

We perform several simulations (see Append.~\ref{app:numerics} for details) of von Neumann measurement schemes, to study how different types of schemes perform in information extraction. The quantifier of this performance will be observational entropy. Specifically, we show how good these indirect measurements are in comparison with direct measurements and with each other. This depends on several model parameters such as the amount of localization of the auxiliary system, the number of indirect measurements performed, and the number of contacts between the auxiliary and the measured system.

Due to the computational complexity of the given task, our analysis is restricted to a two-dimensional system. However, these simulations are mainly a demonstration of the general theory, proven for systems of arbitrary dimension.

\subsection{General parameters of our model}\label{sec:genpar}
The computational basis will be defined by the eigenbasis of the measurement operator,
\[\label{eq:M}
\hat{M}=
\begin{pmatrix}
0 & 0 \\
0 & 2 
\end{pmatrix},
\]
whose eigenvalues are ${\mu_1=0}$ and ${\mu_2=2}$. We evolve the system with two types of Hamiltonian: one that commutes with the measurement operator $\hat{M}$,
\[\label{eq:yesH}
\hat{H}=
\begin{pmatrix}
0 & 0 \\
0 & 2 
\end{pmatrix},
\]
and another that does not,
\[\label{eq:nonH}
\hat{H}=\begin{pmatrix}
0 & 1+i \\
1-i & 2 
\end{pmatrix}.
\]
The exact form of the commuting Hamiltonian is irrelevant. Due to $[\hat{H},\hat{M}]=0$, the Hamiltonian will not play a role in the probabilities ($p_{x}$ and $p_\bx$) and volumes ($V_{x}$ and $V_\bx$), thereby having no effect on the value of observational entropy.

We parametrize the initial qubit state as
\[\label{eq:initial_state}
\R(\phi,\theta,\alpha)=\U
\begin{pmatrix}
\cos^2\alpha & 0 \\
0 & \sin^2\alpha 
\end{pmatrix}
\U^\dag,
\]
where
\[
\U=\begin{pmatrix}
e^{i\phi} & 0 \\
0 & e^{-i\phi} 
\end{pmatrix}
\begin{pmatrix}
\cos \theta & \sin \theta \\
-\sin \theta & \cos \theta 
\end{pmatrix}.
\]
Moreover, to fix the maximum of observational entropy, we chose the basis of the logarithm for the observational entropy to be $\log_2$ instead of the natural logarithm in all of the examples (this changes only the scale).

The results are shown in the following figures, whose captions contain additional discussion: Figure~\ref{Fig:single}~depicts observational entropy of a single measurement scheme (Sec.~\ref{sec:singlevN}). Figures~\ref{Fig:non} and~\ref{Fig:int} show the case of repeated measurements (Sec.~\ref{sec:rep_measurements}) and repeated contacts (Sec.~\ref{sec:rep_contacts}) schemes, for an evolution Hamiltonian that does and does not commute with the measurement operator, respectively.
Figure~\ref{Fig:limit} depicts a different type of repeated contacts limit, wherein the ratio of interaction time and free evolution time $\kappa/\Delta t=\lambda \Delta \tau/\Delta t$ are fixed [see Eq.~\eqref{eq:limit_PO}], such that the large $N$ limit can be taken unlike previous figures.

\subsection{Relation to distinguishability of the observed and expected probabilities}
Panels (b) and (c) of Figs.~\ref{Fig:non} and~\ref{Fig:int} depict the evolution of the quantum system via the probabilities of position measurements and normalized volumes respectively. They also shed light on the behavior of observation entropy through the relation
\[\label{eq:KLdivergence}
S_\bC=\ln V-D_{\mathrm{KL}}(p_\bi\|V_\bi/V),
%S_\bC=\ln V-D_{\mathrm{KL}}(p(x)||V(x)/V)
\]
where $V=\dim \HS$ is the dimension of the Hilbert space. This identity shows that observational entropy can be expressed in terms of Kullback-Leibler (KL) divergence, defined as
\[
D_{\mathrm{KL}}(p_i\|q_i)=\sum_ip_i\ln \frac{p_i}{q_i}.
\]
KL divergence is a measure of distinguishability of two classical probability distributions $\{p_i\}$ and $\{q_i\}$. Thus, we can interpret observational entropy as a measure of distinguishability between the observed probability $\{p_\bi\}$, and the probability of outcomes $\{V_\bi/V=\tr[\PO_\bi(\R_{\max})]\}$. The second probability is the one that would be produced if one were to measure on the maximally uncertain state $\R_{\max}=\I/V$, i.e., a quantum system that we have absolutely no information about. Thus the probability distribution produced by this state serves as our reference, and the KL divergence thus shows how far are we from having absolutely no information.

\subsection{Comparison of repeated measurements and repeated contacts}
Finally, let us comment on different scalings with which the repeated measurements and repeated contacts approach the projective entropy in Fig.~\ref{Fig:non}. In Append.~\ref{app:explicit_forms} we show that both $p$ and $V$ are a linear combination of Gaussians with peak positions given by eigenvalues of the measurement operator $\hat{M}$. In the repeated contact scheme, the peaks of the furthest Gaussians are $N(\mu_{\max}-\mu_{\min})$ far from each other. On the other hand, in repeated measurements the peaks form corners of a hypercube with $(\mu_{\min},\dots,\mu_{\min})$ and $(\mu_{\max},\dots,\mu_{\max})$ being the furthest points. Hence the Euclidean distance gives the result $\sqrt{N}(\mu_{\max}-\mu_{\min})$ for the distance between them. Relation~\eqref{eq:KLdivergence} tells us that the minimum of observational entropy is reached when the difference between distributions $p$ and $V$ is maximal and their overlap is minimized. Since the the standard deviation $\Omega$ of these Gaussians is fixed and does not change with $N$, to minimize the overlap one would ideally create a distribution $p$ with a peak that is the furthest possible distance from the peak of $V$. Since the scaling of the distance of the furthest peaks differs by a $\sqrt{N}$ we would expect that also the speed of convergence scales with the same factor. However, this argument does not represent proof, which we leave for future work.

\section{Observational entropy as a tool for quantum state inference}\label{sec:quantum_tomography}
In this section, we describe an algorithm through which we can infer the state of the quantum system, when observational entropy becomes equal to the von Neumann entropy. We will show that it is possible to use an indirect measurement scheme to determine the state of the system when
\begin{enumerate}
    \item we know the von Neumann entropy and
    \item we manage to find a measurement such that its corresponding observational entropy is equal to the von Neumann entropy,
\end{enumerate}
with two additional implicit assumptions, namely,
\begin{enumerate}
\setcounter{enumi}{2}
    \item we know probabilities $p_\bi$, for example, because we experimentally determined them and
    \item we know the coarse-grainings that defines the observational entropy.
\end{enumerate}
This finding holds for any type of measurement scheme (projective, indirect, or any other generalized measurement scheme). To show that, we will provide an explicit algorithm of \emph{how} to determine the quantum state when these conditions are satisfied. This \emph{quantum state inference} algorithm is based on the equality condition of Theorem~\ref{thm:bounded_multiple}, Eq.~\eqref{eq:rhoandPi} specifically.

\subsection{Algorithm for the quantum state inference}
Let us denote the coarse-graining for which observational entropy is equal to von Neumann entropy as $\C=\{\PO_{\bi}\}$ and its corresponding set of POVM elements as $\{\hPi_{\bi}\}$. We will also denote the index set of unused indices in the algorithm as $I_{\mathrm{un}}$. The algorithm goes as follows:
\begin{enumerate}
\setcounter{enumi}{-1}
    \item Initialize $\R=0$ (a null matrix at the beginning which will represent the inferred quantum state at the end of the algorithm). Initialize the set of unused indices (measurement outcomes) as the set of all indices, $I_{\mathrm{un}}=\{\bi\}$.
    \item Take an index $\bi\in I_{\mathrm{un}}$ that has not been used before in this algorithm. Are there any? (In other words, is the set $I_{\mathrm{un}}$ non-empty?)
    \begin{enumerate}
        \item YES: Update the set of unused indices by subtracting that index, $I_{\mathrm{un}}=I_{\mathrm{un}}\setminus \bi$. Continue to 2.
        \item NO: Return $\R$ (this is the inferred state of the system).
    \end{enumerate}
    \item Initialize $\P=\hPi_{\bi}$ and set (or reset) $I=\{\bi\}$ (a set of indexes that defines $\P$).%(what will become a projector projecting into an eigenspace of $\R$) 
    \item Find all $\bi'$ that have not been used before such that $\hPi_{\bi'}\P\neq 0$. Are there any?
    \begin{enumerate}
        \item YES: Update $\P=\P+\sum_{\bi'}\hPi_{\bi'}$ and $I=I\cup\{\bi'\}_{\bi'}$. Update the unused index set by subtracting all those indices as $I_{\mathrm{un}}=I_{\mathrm{un}}\setminus \{\bi'\}_{\bi'}$. Go back to 3.
        \item NO: According to Lemma~\ref{lma:eigenprojector}, $\P$ is a projector. Continue to 4.
    \end{enumerate}
    \item Calculate $\rho=(\sum_{\bi\in I}p_\bi)/\tr[\P]$ (where $p_\bi$'s denote the probabilities of outcomes). Accodring the Lemma~\ref{lma:eigenvalue}, $\rho$ is an eigenvalue.
    \item Update $\R=\R+\rho\P$. Go back to 1.
\end{enumerate}
The fact that $\R$ generated this way is the state of the system comes from the following lemmas.

\begin{lemma}\label{lma:eigenprojector}
$\P$'s generated by steps 1--3 form a complete set of orthogonal projectors, and each $\P$ projects into an eigenspace of the state of the system.
\end{lemma}

\begin{lemma}\label{lma:eigenvalue}
$\rho$ generated by step 4 is an eigenvalue of the state of the system corresponding to projector $\P$.
\end{lemma}
The proofs of these lemmas, which are a consequence of Eq.~\eqref{eq:rhoandPi}, can be found in Append.~\ref{app:quantum_state_inference}. There we also show how this algorithm works on an example of a single von Neumann measurement scheme from Sec.~\ref{sec:singlevN}.

\subsection{Comparison with quantum tomography}
The task of quantum tomography~\cite{nielsen2010quantum,Toninelli2019concepts} is to determine the state of the system by making measurements in many different bases (while assuming we have access to an infinite number of copies of the state of the system). Because a single measurement provides at most $N$ real parameters, but the density matrix depends on $N^2-1$ real parameters, the number of different bases that needs to be employed to achieve this task is also $N$. If these bases are chosen properly, this always leads to the identification of the quantum state.

Our inference algorithm requires just one measurement basis in principle, but we have to be \emph{lucky} to find exactly the optimal measurement that diagonalizes the density matrix. However, this cannot be consistently done, at least not without some additional optimization, because the density matrix is not apriori known. This brings us to the limitations of the presented scheme.

\subsection{Limitations of the presented scheme}
The main limitations of the present method lie in its two inherent assumptions: First, in order to use the algorithm, we need to know the value of the von Neumann entropy. This might be a problem for a system about which we have absolutely no information, and therefore also no knowledge of the von Neumann entropy. In such a situation, it cannot be determined whether the observational entropy has reached its minimum or not. In other words, we cannot know whether there is a better, yet undiscovered coarse-graining leading to its lower value, closer or equal to the unknown von Neumann entropy. 

However, in some situations, this might not be an issue. For example, in cases in which we know the initial state and the system that evolves unitarily. If the unitary evolution operator is unknown, quantum mechanics provides no means to estimate the state of the system after its evolution. But in such cases, since the von Neumann entropy does not evolve, the algorithm can be readily applied.

Second, we need to find the measurement that will lead the observational entropy to reach its minimum given by the von Neumann entropy. Yet, the algorithm itself does not give a recipe on how to do that. One could, for example, randomly try different types of measurements and pick a few of them that give a very low entropy. They could then further minimize the corresponding observational entropy using optimization over local parameters. The method of finding this optimal measurement should employ adaptive feedback, similar to those developed for quantum tomography~\cite{adaptive2016,adaptive2017,adaptive2017b,adaptive2018,adaptive2021}, in which the next measurement is determined by the outcome of the previous measurements, so that in the limit of a large number of measurements the minimum of observational entropy (given by von Neumann entropy) is reached. This is akin to a global optimization problem that requires adaptive feedback on the set of measurements, which we leave for future work.

Finally, it is an interesting open problem whether the aforementioned algorithm could be generalized to situations in which the von Neumann entropy is not known. Or, to situations when it is known but observational entropy after a finite set of measurements is only approximately equal to the von Neumann entropy. In such cases, it should be possible to provide an error estimate for the density matrix that should tend to zero when the observational entropy converges to the von Neumann entropy.

\section{Discussion, Conclusion, and Future Directions}\label{sec:conclusion}

In this paper, we generalized the concept of observational entropy to include general coarse-grainings (given by generalized measurements). This was motivated by a growing interest from the experimentalists in generalized quantum measurements and to unravel the fundamental nature of observational entropy and its interpretation. In this quest, we overcame and resolved several important subtleties. These were answering how to define a general coarse-graining, how to treat POVM elements that are not orthogonal with each other, what is an appropriate definition of volume of a macrostate, and how can we compare different coarse-grainigs.

The main message of this paper is that even with the general definition of observational entropy, which is defined by a series of possibly non-commuting, generalized measurements, all of the important properties still hold. Observational entropy can be therefore still interpreted as a measure of uncertainty that an observer performing a series of measurements would associate with the initial state of the system. The properties can be summarized as follows: observational entropy is lower bounded by the minimal uncertainty given by von Neumann entropy, upper bounded by the maximal entropy given by the logarithm of the system dimension. With each additional coarse-graining, observational entropy cannot increase, expressing that \emph{``each additional measurement can only increase the observers' knowledge of the state of the system.''} If one coarse-graining is finer than the other, then observational entropy is always lower for a finer coarse-graining, implying that \emph{``an observer that makes a more precise measurement will get to know at least as much as an observer that makes an imprecise one.''}

We applied this concept to study how indirect measurements perform in information extraction as compared to a direct measurement. Performing an analysis of a general scenario of indirect measurements we found, for example, that any pure state of the system can be perfectly indirectly determined by using a two-level auxiliary system---a result which a posteriori seems very natural~\cite{PELLONPAA2017}. To illustrate the application on a specific and timely example, we applied this concept to various von Neumann measurement schemes, in which a quantum system is probed through an auxiliary system consisting of a classical particle. Since observational entropy measures information extracted in different situations, it serves as a performance quantifier for the different measurement schemes. This not only provides new insights into the understanding of the various measurement schemes, but also determines which one is the best given a situation.

Moreover, computing observational entropy is a relatively simple, yet a powerful test on the performance of any sequence of measurements in information extraction. For example, the construction of a quantum computer, which is expected to solve complex problems beyond the capabilities of a classical computer, will require a quantum memory, and a mechanism that reads the output of the computation~\cite{ladd2010quantum,national2019quantum}. Computing observational entropy can determine which are the least invasive measurement schemes to read out both the quantum memory and the computation output.

Finally, we showed that observational entropy can serve as a tool for quantum state identification. We did that by showing that the knowledge of the coarse-graining that leads to the minimum of the observational entropy allows for a successful inference of the state of the system. We presented a general algorithm that achieves this goal to minimized observational entropy. Generalizing this connection to situations when the knowledge is not perfect, for example, for the cases in which the coarse-graining gives a low, but not the minimal value of observational entropy, provides a new exciting direction of study.

An application of this theory could also lie in the study of microscopic thermodynamic systems, and especially in the understanding quantum entropy production~\cite{Goold2016role,Vinjanampathy2016quantum,binder2019thermodynamics}. While the groundwork of using observational entropy for these purposes was already established in Refs.~\cite{strasberg2021first,riera2021quantum}, these works are limited to the use of projective measurements, which meant that the system and the bath had to be considered together as a whole. The framework of observational entropy could also help in answering critical questions of the impact of finite baths~\cite{Thingna17,marcantoni2017entropy,Thingna19,riera2021quantum,martins2021nonmarkovianity}, such as \emph{``how much information about the system can be extracted by measuring a bath,''} that requires generalized (non-projective) measurements on the bath.

Observational entropy has already been used to study black holes~\cite{schindler2021unitarity} and big bang~\cite{lautaro2021quantum}. In black holes, a common question concerns the amount of information lost or gained through evaporation~\cite{Almheiri2021}. Measuring the evaporated states that escape the black hole provides some information on its inner structure. However, performing a projective measurement on these states mathematically corresponds to performing a generalized measurement on the black hole instead. The framework developed here allows for computing exactly how much information on the black hole has been gained by doing this measurement.

There are also many potential applications of this quantity: it can be used in every scenario that involves information gain while making a quantum measurement. We hope that the future development of this framework, as well as finding applications by experts in their respective fields, will demonstrate its practicality beyond its current status.

\begin{acknowledgments}
D.\v S.~thanks Anthony Aguirre and various listeners of his talks for asking the right questions, which sparked this research. We thank Joseph C. Schindler for very interesting discussions which, among other things, led to the current definition of a finer vector of coarse-graining, for his careful reading of this manuscript including the proofs, and the feedback received. This research was supported by the Foundational Questions Institute (FQXi.org),  the Faggin Presidential Chair Fund, and the Institute for Basic Science in South Korea (IBS-R024-Y2
 and IBS-R024-D1).
\end{acknowledgments}

\appendix

\section{Historical overview}\label{app:History}
The history of observational entropy goes all the way back to John von Neumann. It first appeared in his paper~\cite{vonNeumann1929} in 1929, where he motivated the introduction of this entropy by criticizing the quantity which we know today as the von Neumann entropy. He said: \emph{``The expressions for [the von Neumann] entropy given by the author in~\cite{Neumann1927thermodynamik} are not applicable here in the way they were intended, as they were computed from the perspective of an observer who can carry out all measurements that are possible in principle—i.e., regardless of whether they are macroscopic (for example, there every pure state has entropy 0, only mixtures have entropies greater than 0!).''} He pointed out that von Neumann entropy cannot represent a good generalization of thermodynamic entropy, since any pure state, even those at high energies, would have associated zero entropy. Additionally, von Neumann entropy remains constant in an isolated quantum system, which would suggest that in such a system, left to spontaneous evolution, no information is lost. However, this is again for an observer that can perform \emph{all} measurements, even those that are in a very complicated/highly entangled basis. On the other hand, a realistic observer with limited capabilities or resources will always observe an increase in entropy (see also Ref.~\cite{sheridan2020man}).

%In 1929, in the same paper~\cite{vonNeumann1929} in which von Neumann introduced his generalization of Boltzmann H-theorem to quantum physics, he also made less known, yet not less remarkable observation. He criticised the quantity which we know today as the von Neumann entropy, saying: \emph{``The expressions for [the von Neumann] entropy given by the author in~\cite{Neumann1927thermodynamik} are not applicable here in the way they were intended, as they were computed from the perspective of an observer who can carry out all measurements that are possible in principle—i.e., regardless of whether they are macroscopic (for example, there every pure state has entropy 0, only mixtures have entropies greater than 0!).''} He pointed out that von Neumann entropy cannot represent a good generalization of thermodynamic entropy, since any pure state, even those at high energies, would have associated zero entropy. Additionally, von Neumann entropy remains constant in an isolated quantum system, which would suggest than in such a system, left to spontaneous evolution, no information is lost---however, that is again only for an observer that can perform all measurements in principle, even those that are in a very complicated/highly entangled basis. A realistic observer with limited capabilities will always observe an increase in entropy.

While referring to a discussion with Eugene Wigner\footnote{Von Neumann mentions that this definition was E. Wigner's idea, while saying that there is no need to go into a general theory. We searched for a follow-up paper by E. Wigner, where this general theory would be discussed, but could not find any---it seems that his interests at the time took him elsewhere---to develop the theory of symmetries in quantum mechanics, and then apply it to derive essential properties of the nuclei, for which he received 1963 Nobel Prize in Physics~\cite{wigner1963nobel}.}, John von Neumann proposed an alternative quantity to rectify this unsatisfactory behavior,
\[\label{eq:vNOE}
S(\psi)=-\sum_E \bra{\psi}\P_E\ket{\psi}\ln \frac{\bra{\psi}\P_E\ket{\psi}}{V_E},
\]
which does not suffer from the same drawbacks. Above, $\P_E$ is a projector on an energy subspace (surface), and $V_E$ is the number of orbits (microstates) in an energy surface. Here, $\bra{\psi}\P_E\ket{\psi} = p_E$ is the probability of finding the system in an energy shell of energy $E$. This entropy measures the lack of information due to an observer's limited capability of distinguishing energy eigenstates within a small energy gap $\Delta E$. The above is a natural generalization of the Boltzmann entropy\footnote{Consider a point in phase-space that belongs to an energy surface $E$. For this phase-space point, the original Boltzmann definition (which we call microcanonical/surface entropy here) reads $S_{\mathrm{B}}=\ln V_E$. Also the von Neumann's definition [Eq.~\eqref{eq:vNOE}], if a quantum state belongs into an energy subspace $E$, will have $p_E=1$ and therefore it gives the same value $S(\psi)=\ln V_E$.} and has found several applications~\cite{percival1961almost,wehrl1978general,tolman1979principles,penrose1979foundations,zubarev1996statistical,Latora1999kolmogorov,nauenberg2004evolution,ohya2004quantum,gemmer2014entropy} with the first one being the quantum generalization of the celebrated Boltzmann's H-theorem~\cite{vonNeumann1929}.

The idea of von Neumann's alternative entropy [Eq.~\eqref{eq:vNOE}] was recently revived~\cite{safranek2019short,safranek2019long}, and generalized to include multiple non-commuting coarse-grainings. It was also named \emph{observational entropy}, due to the fact that it is an observer's ability to measure a certain macroscopic variable that determines the coarse-graining.

\section{Proofs of Theorems}\label{app:proofs}
The proofs of Theorems~\ref{thm:bounded_multiple}.~and \ref{thm:non-increase}.~are a modification of those done for projective measurements (Thms.~7.~and 8.~of~\cite{safranek2019long}), and the spirit of the proof is exactly the same. The proof of Theorem~\ref{thm:monotonic} is similar to the proof of Theorem~2.~in~\cite{safranek2019long}), but modified more significantly, because it uses a more general (and different-looking, although being equivalent on the special cases) Definition~\ref{def:finer_set_coarse_graining}.

All inequalities follow from the application of the well-known theorem:
\begin{theorem}(Jensen, see, e.g., Refs.~\cite{jensen1906fonctions,chandler1987introduction,needham1993visual})\label{thm:Jensen}
Let $f$ be a strictly concave function, $0\leq a_i \leq 1$, $\sum_i a_i=1$. Then for any $b_i\in\mathbb{R}$,
\[
f\big(\sum_i a_i b_i\big)\geq \sum_i a_i f(b_i).
\]
$f(\sum_i a_i b_i)= \sum_i a_i f(b_i)$ if and only if
$
(\forall i,j|a_i\neq 0, a_j\neq 0)(b_i=b_j).
$
\end{theorem}

\subsection{Proof of Theorem~\ref{thm:bounded_multiple}}
\begin{proof}
In this subsection we are going to prove $S_{\mathrm{vN}}\leq S_{\bC}$ and $S_{\bC}\leq \ln \dim \HS$, each with its equality condition. The vector of coarse-graining 
\[
\bC=(\C_1,\dots,\C_n)=\{\PO_\bi\},
\]
and
\[
\PO_{\bi}(\R)=\sum_{\bm}\K_{\bi\bm}\R\K_{\bi\bm}^\dag,
\]
with the Krauss operators $\K_{\bi\bm}=\K_{i_n m_n}\cdots \K_{i_1 m_1}$ and $\K_{\bi\bm}^\dag=\K_{i_1 m_1}^\dag\cdots \K_{i_n m_n}^\dag$. The vector of outcomes is $\bi=(i_1,\dots,i_n)$, and $\bm=(m_1,\dots,m_n)$. We define the POVM element
\[\label{eq:Kraussdef}
\hPi_\bi=\sum_\bm\K_{\bi\bm}^\dag\K_{\bi\bm}.
\] 
From the definition of coarse-graining, we also have 
\[\label{eq:completeness2}
\sum_{\bi}\hPi_\bi=\I.
\]
In these equations, to keep the notation short, we write
\[\label{eq:denotingim}
\sum_{i_1,\dots,i_n}\sum_{m_1,\dots,m_n}\equiv \sum_{\bi}\sum_{\bm}\equiv\sum_{\bi,\bm}.
\]

We denote the spectral decomposition of the density matrix as $\R=\sum_x\rho_x\pro{x}{x}$, where $\ket{x}$ denotes the eigenvector of the density matrix and $\rho_x$ is the corresponding eigenvalue. The eigenvalues are not necessarily different for different $x$, thus making this decomposition not unique. We also denote the unique decomposition of the density matrix in terms of its eigenprojectors $\R=\sum_{\rho}\rho\P_{\rho}$, where eigenvalues $\rho$ are now different from each other. It follows that for each $x$ there is $\rho$ such that $\rho_x=\rho$. 

Now we prove $S_{\mathrm{vN}}\leq S_{\bC}$ together with its equality condition. We begin by defining
\[
a_x^{(\bi)}\equiv\frac{\bra{x}\hPi_\bi\ket{x}}{V_\bi}
\]
for $V_\bi\neq 0$ and $a_x^{(\bi)}\equiv 0$ for $V_\bi= 0$. Using the spectral decomposition of $\R$ we have
\[\label{eq:p_over_V}
\begin{split}
\frac{p_\bi}{V_\bi}&=\frac{\sum_x \rho_x \bra{x}\hPi_\bi\ket{x}}{V_\bi}\\
&=\sum_x \rho_x a_x^{(\bi)}.
\end{split}
\]
The cyclicity of the trace dictates $V_\bi=\tr[\hPi_\bi]=\sum_x\bra{x}\hPi_\bi\ket{x}$, from which follows
\[\label{eq:sum_ax}
\sum_xa_x^{(\bi)}=1.
\]
Applying the completeness relation~\eqref{eq:completeness2}, we also have
\[\label{sum_V_ax}
\sum_\bi V_\bi a_x^{(\bi)}=\sum_{\bi}\bra{x}\hPi_\bi\ket{x}=\braket{x}{x}=1.
\]

A series of identities and inequalities follow
\[
\begin{split}
S_{\bC}&=-\sum_{\bi}p_\bi\ln \frac{p_\bi}{V_{\bi}}\\
&=-\sum_{\bi}V_{\bi}\frac{p_\bi}{V_{\bi}}\ln \frac{p_\bi}{V_{\bi}}\\
&=\sum_{\bi}V_{\bi}\left(-\sum_x \rho_x a_x^{(\bi)}\ln \sum_x \rho_x a_x^{(\bi)}\right)\\
&\geq \sum_{\bi}V_{\bi} \left(-\sum_x a_x^{(\bi)} \rho_x\ln \rho_x\right)\\
&=-\sum_x\left( \sum_{\bi}V_{\bi} a_x^{(\bi)}\right) \rho_x\ln \rho_x=S_{\mathrm{vN}}.
\end{split}
\]
The third identity follows from Eq.~\eqref{eq:p_over_V}, and the last identity follows from Eq.~\eqref{sum_V_ax}. We have used the Jensen's Theorem~\ref{thm:Jensen} on the strictly concave function $f(x)=-x\ln x$ in order to obtain the inequality. We have chosen $a_x\equiv a_x^{(\bi)}$ and $b_x=\rho_x$ for the Theorem. This is a valid choice due to $0\leq a_x^{(\bi)}\leq 1$ and due to Eq.~\eqref{eq:sum_ax}. This proves the inequality $S_{\mathrm{vN}}\leq S_{\bC}$.

According to Jensen's Theorem, this inequality becomes identity if and only if
\begin{eqnarray}
    \label{eq:first_eq_condition_old_2}
    \rho_x=\rho_{\tilde{x}} \quad (&&\forall \bi ~\& \\ && \forall x, \tilde{x}|\bra{x}\hPi_\bi\ket{x}\!\neq\! 0,\bra{\tilde{x}}\hPi_\bi\ket{\tilde{x}}\!\neq\! 0). \nonumber 
\end{eqnarray}
Having $\bra{x}\hPi_{\bi}\ket{x}=0$ is equivalent to $\hPi_{\bi}\pro{x}{x}=0$. In order to show this, we begin by considering $\bra{x}\hPi_{\bi}\ket{x}=0$ and inserting Eq.~\eqref{eq:Kraussdef} to obtain,
\begin{eqnarray}
\label{eq:seriesimpli}
&&\bra{x}\sum_\bm\K_{\bi\bm}^\dag\K_{\bi\bm}\ket{x}=0\nonumber \\
&\Rightarrow& \sum_\bm\norm{\K_{\bi\bm}\ket{x}}^2=0\\
&\Rightarrow& \K_{\bi\bm}\ket{x}=0 \quad \quad \quad\, \quad \quad\quad\quad\quad \quad \quad \quad \quad (\forall \bm)\nonumber \\
&\Rightarrow& \K_{\bi\bm}^\dag\K_{\bi\bm}\pro{x}{x}=0 \quad\quad \quad \quad \quad \quad \quad\quad \quad  (\forall \bm)\nonumber \\
&\Rightarrow& \sum_\bm\K_{\bi\bm}^\dag\K_{\bi\bm}\pro{x}{x}=\hPi_\bi\pro{x}{x}=0.\nonumber 
\end{eqnarray}
Moreover, trivially, $\hPi_\bi\pro{x}{x}=0$ implies $\bra{x}\hPi_\bi\ket{x}=0$. This means that $\bra{x}\hPi_{\bi}\ket{x}=0\Leftrightarrow\hPi_{\bi}\pro{x}{x}=0$ and also 
$\bra{x}\hPi_{\bi}\ket{x}\neq 0\Leftrightarrow\hPi_{\bi}\pro{x}{x}\neq 0$. Using this equivalence, we rewrite the condition for the inequality to be the identity, Eq.~\eqref{eq:first_eq_condition_old_2}, as
\begin{eqnarray}
    \label{eq:first_eq_condition_old_3}
    \rho_x=\rho_{\tilde{x}} \quad (&&\forall \bi ~\& \\ &&\forall x, \tilde{x}|\,\hPi_\bi\pro{x}{x}\!\neq\! 0,\hPi_\bi\pro{\tilde{x}}{\tilde{x}}\!\neq\! 0). \nonumber 
\end{eqnarray}

We explain this condition as follows: the inequality ($S_{\mathrm{vN}}\leq S_{\bC}$) becomes identity ($S_{\mathrm{vN}} = S_{\bC}$) when for a fixed multi-index $\bi$, all eigenvectors of the density matrix $\ket{x}$ such that $\hPi_{\bi}\pro{x}{x}\neq 0$ must have the same associated eigenvalue $\rho_x$. In other words, this unique eigenvalue can be associated to the multi-index $\bi$ itself, so we can relabel $\rho_\bi \equiv\rho_x$, where the eigenvalue $\rho_x$ is given by any representative $x$ such that $\hPi_{\bi}\pro{x}{x}\neq 0$. This must hold for every multi-index $\bi$, in order for the inequality to become an identity. Therefore, this defines a unique map that associates some eigenvalue of the density matrix to each multi-index $\bi$. To highlight the existence of this map we can extend Eq.~\eqref{eq:first_eq_condition_old_3} and write
\begin{eqnarray}
    \label{eq:first_eq_condition_b}
    \rho_x=\rho_{\tilde{x}}\equiv\rho_\bi \quad (&&\forall \bi ~\& \\ 
    &&\forall x, \tilde{x}|\,\hPi_\bi\pro{x}{x}\!\neq\! 0,\hPi_\bi\pro{\tilde{x}}{\tilde{x}}\!\neq\! 0). \nonumber
\end{eqnarray}
Then, defining a set
\[
I^{(\bi)}=\{x|\rho_x=\rho_\bi\},
\]
using condition~\eqref{eq:first_eq_condition_b}, and $\sum_x\ket{x}\bra{x}=\I$, we obtain
\[\label{eq:nonzerooni}
\hPi_\bi=\hPi_\bi\!\sum_x\!\pro{x}{x}=\hPi_\bi\!\sum_{x\in I^{(\bi)}}\!\!\pro{x}{x}=\hPi_\bi\P_{\rho_\bi}.
\]
Assuming that $\rho\neq \rho_\bi$, we can multiply this equation by $\P_\rho$ from the right, and from the orthogonality of projectors we obtain
\[
\hPi_\bi\P_\rho=0.
\]
In other words, this means that for every $\bi'$ such that $\rho_{\bi'}\neq\rho_{\bi}$, 
\[\label{eq:zerooni}
\hPi_{\bi'}\P_{\rho_\bi}=0.
\]
Finally, for any eigenvalue $\rho$ we define an index set
\[
I^{(\rho)}=\{\bi|\rho_\bi=\rho\}.
\]
Using the completeness relation $\sum_\bi\hPi_\bi=\I$, and combining Eqs.~\eqref{eq:nonzerooni} and \eqref{eq:zerooni} gives
\[\label{eq:finalfinercg}
\P_\rho=\sum_\bi\hPi_\bi\P_\rho=\sum_{\bi\in I^{(\rho)}}\hPi_\bi\P_\rho=\sum_{\bi\in I^{(\rho)}}\hPi_\bi,
\]
which by definition means that $\C_\R\hookrightarrow\bC$. 

Conversely, for a contradiction we assume Eq.~\eqref{eq:finalfinercg} holds, but Eq.~\eqref{eq:first_eq_condition_old_2} does not, which would mean that there are $x$ and $\tilde x$ such that $\bra{x}\hPi_\bi\ket{x}\!\neq\! 0$ and $\bra{\tilde{x}}\hPi_\bi\ket{\tilde{x}}\!\neq\! 0$ while $\rho_x\neq \rho_{\tilde{x}}$. We assume arbitrary $\bi$ and if $\bi\in\I^{(\rho)}$ where $\rho\neq \rho_x$ then multiplying Eq.~\eqref{eq:finalfinercg} by $\P_{\rho_x}$ gives
\[
0=\bra{x}\P_\rho\P_{\rho_x}\ket{x}=\bra{x}\P_\rho\ket{x}=\sum_{\bi\in \I^{(\rho)}}\bra{x}\hPi_\bi\ket{x},
\]
which implies that for every $\bi\in \I^{(\rho)}$, $\bra{x}\hPi_\bi\ket{x}=0$, due to positivity of operators $\hPi_\bi$. Thus, if $\ket{x}$ and $\ket{\tilde x}$ are associated with different eigenvalues, at least one of the $\bra{x}\hPi_\bi\ket{x}$ or $\bra{\tilde x}\hPi_\bi\ket{\tilde x}$ must be zero. This is a contradiction with our assumption. Thus, Eq.~\eqref{eq:finalfinercg} implies \eqref{eq:first_eq_condition_old_2}, which is equivalent to the equality condition $S_{\mathrm{vN}}= S_{\bC}$. We have therefore shown that $S_{\mathrm{vN}}= S_{\bC}$ if and only if $\C_\R\hookrightarrow\bC$, which concludes the first part of the proof.

Next, we prove $S_{\bC}\leq \ln \dim \HS$ together with its equality condition.
\[
\begin{split}
&S_{\bC}=\sum_{\bi:p_\bi\neq 0}p_\bi\ln \frac{V_{\bi}}{p_\bi}\leq \ln\left(\sum_{\bi:p_\bi\neq 0}p_\bi\frac{V_{\bi}}{p_\bi}\right)\\
&\leq \ln\left(\sum_{\bi}V_{\bi}\right)=\ln \tr{\I}=\ln \dim\; \!\HS.
\end{split}
\]
In here, the first inequality follows from the Jensen's Theorem, which was applied on strictly concave function $f(x)=\ln x$, while we have chosen $a_\bi\equiv p_\bi$ and $b_\bi\equiv V_{\bi}/p_\bi$ for the Theorem. This is a valid choice, because $0\leq a_\bi\leq 1$ and $\sum_\bi a_\bi=1$ hold. The second inequality follows from $V_i\geq 0$ while realizing that the logarithm is an increasing function. The second identity follows from the completeness relation $\sum_{\bi,\bm} \K_{\bi\bm}^\dag\K_{\bi\bm}=\I$ and from the definition of $V_\bi$.

The first inequality becomes equality if and only if
\[
\frac{V_\bi}{p_\bi}=\frac{V_{\boldsymbol{\bi'}}}{p_{\boldsymbol{\bi'}}}=c \quad\quad  (\forall \bi, \bi'| p_\bi\neq 0, p_{\boldsymbol{\bi'}}\neq 0)
\]
for some real constant $c$. In order to determine the value of this constant, we express this condition as $V_\bi=c p_\bi$ and then we sum over all multi-indexes $\bi$ for which $p_\bi\neq 0$. This leads to $c=\sum_{\bi:p_\bi\neq 0}V_\bi$. Therefore, we can write the first equality condition as
\[
p_\bi=\frac{V_\bi}{\sum_{\bi:p_\bi\neq 0}V_\bi} \quad\quad \quad (\forall p_\bi\neq 0).
\]
Since logarithm is a strictly increasing function, the second inequality becomes identity if and only if $p_\bi=0$ and $V_\bi=0$, for every multi-index $\bi$. If this is true, then also  $\sum_{\bi:p_\bi\neq 0}V_\bi=\sum_{\bi}V_\bi=\dim \HS$ holds, where we have used  the definition of $V_\bi$ and the completeness relation $\sum_{\bi,\bm} \K_{\bi\bm}^\dag\K_{\bi\bm}=\I$. When we combine both of these equality conditions, we obtain that $S_{\bC}=\ln \dim\; \!\HS$ if and only if
\[
p_\bi=\frac{V_\bi}{\dim \HS} \quad \quad \quad \quad \quad \quad \quad (\forall p_\bi).
\]
This completes the proof.
\end{proof}

\subsection{Proof of Theorem~\ref{thm:non-increase}}\label{sec:proof2}
\begin{proof}
To make our notation easy to read, we denote $p_{i_1,\dots,i_n,i_{n+1}}\equiv p_{\bi,i_{n+1}}$, and $V_{i_1,\dots,i_n,i_{n+1}}\equiv V_{\bi,i_{n+1}}$, while we also use the same notation that we used in the previous proof. Using
\[
p_{\bi}=\sum_{i_{n+1}}p_{\bi,i_{n+1}},\quad
V_{\bi}=\sum_{i_{n+1}}V_{\bi,i_{n+1}},
\]
in combination with Jensen's Theorem~\ref{thm:Jensen} gives
\begin{eqnarray}
S_{\bC}&=&-\sum_{\bi}p_\bi\ln \frac{p_\bi}{V_{\bi}}\nonumber\\
&=&-\sum_{\bi}\sum_{i_{n+1}}p_{\bi,i_{n+1}}\ln \frac{\sum_{i_{n+1}}p_{\bi,i_{n+1}}}{V_{\bi}}\nonumber\\
&=&-\sum_{\bi}V_{\bi}\bigg(\sum_{i_{n+1}}\frac{p_{\bi,i_{n+1}}}{V_{\bi,i_{n+1}}}\frac{V_{\bi,i_{n+1}}}{V_{\bi}}\bigg)\nonumber\\
&&\times\ln \bigg(\sum_{i_{n+1}}\frac{p_{\bi,i_{n+1}}}{V_{\bi,i_{n+1}}}\frac{V_{\bi,i_{n+1}}}{V_{\bi}}\bigg)\nonumber\\
&\geq& \sum_{\bi}V_{\bi}\sum_{i_{n+1}}\frac{V_{\bi,i_{n+1}}}{V_{\bi}}\bigg(-\frac{p_{\bi,i_{n+1}}}{V_{\bi,i_{n+1}}}\ln \frac{p_{\bi,i_{n+1}}}{V_{\bi,i_{n+1}}}\bigg)\nonumber\\
&=&-\sum_{\bi,i_{n+1}}p_{\bi,i_{n+1}}\ln \frac{p_{\bi,i_{n+1}}}{V_{\bi,i_{n+1}}}\nonumber\\
&=&S_{\bC,\C_{n+1}}.
\end{eqnarray}
Above, we have used $f(x)=-x\ln x$, $a_{i_{n+1}}=V_{\bi,i_{n+1}}/V_{\bi}$ and $b_{i_{n+1}}=p_{\bi,i_{n+1}}/V_{\bi,i_{n+1}}$ to obtain the inequality using Jensen's Theorem~\ref{thm:Jensen}.

The condition for Jensen's inequality to become an identity is
\begin{eqnarray}
\label{eq:cond}
&&\frac{p_{\bi,i_{n+1}}}{V_{\bi,i_{n+1}}}=\frac{p_{\bi,i_{n+1}'}}{V_{\bi,i_{n+1}'}}=c^{(\bi)} \\ &&\quad \quad (\forall \bi ~\& ~
\forall i_{n+1},{i}_{n+1}'|V_{\bi,i_{n+1}}\neq 0, V_{\bi,{i}_{n+1}}'\neq 0), \nonumber
\end{eqnarray}
where $c^{(\bi)}$ is some $\bi$-dependent constant, which we can compute using $\sum_{i_{n+1}}p_{\bi,i_{n+1}}=c^{(\bi)}\sum_{i_{n+1}}V_{\bi,i_{n+1}}$, obtaining $c^{(\bi)}=p_{\bi}/V_{\bi}$. This enables us to rewrite Eq.~\eqref{eq:cond} as
\[
p_{\bi,i_{n+1}}=\frac{V_{\bi,i_{n+1}}}{V_{\bi}}p_{\bi} \quad (\forall \bi ~\& ~\forall i_{n+1}|V_{\bi,i_{n+1}}\neq 0).
\]
For every $V_{\bi,i_{n+1}}= 0$ we also have $p_{\bi,i_{n+1}}=0$, from which we trivially obtain $p_{\bi,i_{n+1}}=p_{\bi}V_{\bi,i_{n+1}}/V_{\bi}$. Therefore, this condition can be simplified further, which gives the final result that $S_{\bC}=S_{\bC,\C_{n+1}}$ if and only if
\[\label{eq:condThm2}
p_{\bi,i_{n+1}}=\frac{V_{\bi,i_{n+1}}}{V_{\bi}}p_{\bi} \quad \quad \quad \quad (\forall \bi ~\& ~\forall i_{n+1}).
\]
This completes the proof.

We make two interesting remarks about this condition: Assuming that $p_{\bi}\neq 0$, we can rewrite the above condition as $p(i_{n+1}|\bi)\equiv p_{\bi,i_{n+1}}/p_{\bi}=V_{\bi,i_{n+1}}/V_{\bi}$. This shows that the entropy will not decrease with additional coarse-graining $\C_{n+1}$ if the conditional probability of the outcome $i_{n+1}$ is proportional to the ratio of the macrostate volumes.

As an example, this equality condition is satisfied when the vector of coarse-grainings $\bC=(\C_1,\dots,\C_n)$ projects onto a pure state. In other words, the equality is satisfied when for every density matrix and every vector of outcomes $\bi$ we can write $\pro{\psi_{\bi}}{\psi_{\bi}}=\PO_\bi(\R)/p_{\bi}$ (where it is important to note that the left-hand side does not depend on the density matrix $\R$, even though the right-hand side does). Since this must hold for any density matrix, it also holds for the maximally mixed state $\R_{\mathrm{id}}=\I/\dim \HS$, which in turn yields $\pro{\psi_{\bi}}{\psi_{\bi}}=\PO_\bi(\I)/V_{\bi}$. Therefore we can write
\[
\begin{split}
p_{\bi,i_{n+1}}&=\tr[\PO_{i_{n+1}}(\pro{\psi_{\bi}}{\psi_{\bi}})]p_{\bi}\\
&=\tr[\PO_{i_{n+1}}\left(\frac{\PO_\bi(\I)}{V_{\bi}}\right)]p_{\bi}=\frac{V_{\bi,i_{n+1}}}{V_{\bi}}p_{\bi},
\end{split}
\]
meaning that the condition~\eqref{eq:condThm2} is satisfied, and thus $S_{\bC}=S_{\bC,\C_{n+1}}$.

As a second example, the equality condition is also satisfied is when $\bC^\dag\hookleftarrow \C^{\dag}$, where $\bC^\dag\equiv(\C_n^\dag,\dots,\C_1^\dag)$, $\C_k^\dag\equiv\{\PO_{i_k}^\dag\}$, and $\PO_{i_k}^\dag=\sum_{m_k}\K_{i_km_k}^\dag(\bullet)\K_{i_km_k}$. To simplify our notation we will use $\PO_\bi^\dag=\PO_{i_1}^\dag\cdots\PO_{i_n}^\dag$. This means that for every multi-index $\bi$ there is index $i_{n+1}^{(\bi)}$ such that
$\PO_\bi^\dag\PO_{i_{n+1}^{(\bi)}}^\dag(\I)=\PO_\bi(\I)$ holds, and for every other index ${i_{n+1}}\neq i_{n+1}^{(\bi)}$, $\PO_\bi\PO_{i_{n+1}^{(\bi)}}^\dag(\I)=0$ holds. For ${i_{n+1}}= i_{n+1}^{(\bi)}$  we obtain a series of identities
\begin{eqnarray}
p_{\bi,i_{n+1}^{(\bi)}}&=&\tr[\PO_{i_{n+1}^{(\bi)}}\PO_\bi(\R)]\nonumber\\
&=&\tr[\sum_{\bm,m_{n+1}}\K_{i_{n+1}m_{n+1}}\K_{\bi\bm}\R\K_{\bi\bm}^\dag\K_{i_{n+1}m_{n+1}}^\dag]\nonumber\\
&=&\tr[\sum_{\bm,m_{n+1}}\K_{\bi\bm}^\dag\K_{i_{n+1}m_{n+1}}^\dag\K_{i_{n+1}m_{n+1}}\K_{\bi\bm}\R]\nonumber\\
&=&\tr[\PO_\bi^\dag\PO_{i_{n+1}}^\dag(\I)\R]\nonumber\\
&=&\tr[\PO_\bi^\dag(\I)\R]\nonumber\\
&=&\tr[\hPi_\bi\R]\nonumber\\
&=&\tr[\PO_\bi(\R)]=\frac{V_{\bi,i_{n+1}^{(\bi)}}}{V_{\bi}}p_{\bi},
\end{eqnarray}
where we have used that $V_{\bi,i_{n+1}^{(\bi)}}/V_{\bi}=1$. Further, when ${i_{n+1}}\neq i_{n+1}^{(\bi)}$, then $V_{\bi,i_{n+1}}/V_{\bi}=0$ and thus
\[
\begin{split}
p_{\bi, i_{n+1}}&=\tr[\PO_{i_{n+1}^{(\bi)}}\PO_\bi(\R)]\\
&=0=\frac{V_{\bi, i_{n+1}}}{V_{\bi}}p_{\bi}.
\end{split}
\]
Combining the last two equations, we get
\[
p_{\bi, i_{n+1}}=\frac{V_{\bi, i_{n+1}}}{V_{\bi}}p_{\bi}
\]
for all $i_{n+1}$, proving that also in this case, $S_{\bC}=S_{\bC,\C_{n+1}}$.
\end{proof}

\subsection{Proof of Theorem~\ref{thm:monotonic}}
\begin{proof}
Let $\{\PO_\bi\}=\bC \hookleftarrow \widetilde{\bC}=\{\PO_\bj\}$. Then by definition, for every multi-index $\bj$ there exists an index set $I^{(\bj)}$ such that
\[
\hPi_\bj=\sum_{\bi\in I^{(\bj)}}\hPi_\bi,
\]
where $\hPi_\bi=\sum_\bm\K_{\bi\bm}^\dag\K_{\bi\bm}$ and $\hPi_\bj=\sum_\bm\K_{\bj\bm}^\dag\K_{\bj\bm}$. Thus, we have
\begin{eqnarray}
p_\bj&=&\tr[\PO_\bj(\R)]=\tr[\hPi_\bj\R] \nonumber \\
&=&\!\!\!\sum_{\bi\in I^{(\bj)}}\!\!\tr[\hPi_\bi\R]=\!\!\sum_{\bi\in I^{(\bj)}}\!p_\bi,
\end{eqnarray}
and similarly
\[
V_{\bj}=\sum_{\bi\in I^{(\bj)}}V_\bi.
\]
The inequality then immediately follows as,
\begin{eqnarray}
S_{\widetilde{\bC}}&=&-\sum_{\bj}p_\bj\ln \frac{p_\bj}{V_{\bj}}\nonumber\\
&=&-\sum_{\bj}\sum_{\bi\in I^{(\bj)}}p_{\bi}\ln \frac{\sum_{\bi\in I^{(\bj)}}p_{\bi}}{V_{\bj}}\nonumber \\
&=&-\sum_{\bj}V_{\bj}\bigg(\sum_{\bi\in I^{(\bj)}}\frac{p_{\bi}}{V_{\bi}}\frac{V_{\bi}}{V_{\bj}}\bigg)\ln \bigg(\sum_{\bi\in I^{(\bj)}}\frac{p_{\bi}}{V_{\bi}}\frac{V_{\bi}}{V_{\bj}}\bigg)\nonumber\\
&\geq& \sum_{\bj}V_{\bj}\sum_{\bi\in I^{(\bj)}}\frac{V_{\bi}}{V_{\bj}}\bigg(-\frac{p_{\bi}}{V_{\bi}}\ln \frac{p_{\bi}}{V_{\bi}}\bigg)\nonumber\\
&=&-\sum_{\bi}p_{\bi}\ln \frac{p_{\bi}}{V_{\bi}}=S_{\bC},
\end{eqnarray}
where we have chosen a strictly concave function $f(x)=-x\ln x$, $a_{\bi}=V_{\bi}/V_{\bj}$ and $b_{\bi}=p_{\bi}/V_{\bi}$ for $\bi\in I^{(\bj)}$ for the Jensen's Theorem~\ref{thm:Jensen}.

The equality conditions from the Jensen's inequality show that $S_{\widetilde{\bC}}=S_{\bC}$ if and only if
\[
\frac{p_{\bi}}{V_{\bi}}=\frac{p_{\bi'}}{V_{\bi'}}=c^{(\bj)} \quad \quad (\forall \bj ~ \& ~\forall \bi,\bi'\in I^{(\bj)}).
\]
To determine the constant $c^{(\bj)}$ we multiply the equation by $V_\bi$ and sum over all $\forall \bi\in I^{(\bj)}$, which gives $c^{(\bj)}=p_{\bj}/V_{\bj}$.
Therefore, $S_{\widetilde{\bC}}=S_{\bC}$ if and only if
\[
p_{\bi}=\frac{V_\bi}{V_{\bj}}p_{\bj}\quad\quad \quad \quad (\forall \bj ~\&~\forall \bi\in I^{(\bj)}).
\]
\end{proof}

\section{Equivalence of definitions of finer vector of coarse-grainings}\label{app:equivalence_definitions}

In~\cite{safranek2019long} we defined a finer set of coarse-graining (of projective measurements) in the following way:

\begin{definition}\label{def:finer_set_coarse_graining_old}~(Finer vector of coarse-grainings -- old definition): A vector of coarse-grainings $\bC=(\C_1,\dots,\C_n)$ is finer than coarse-graining $\C\equiv \widetilde\bC$ (and denote $\bC\hookleftarrow \C$), when for every multi-index $\bi=(i_1,\dots,i_n)$ there exists $\P_j\in\C$ such that
\[\label{eq:finer_set_condition_old}
\P_{i_n}\cdots\P_{i_1}\P_j=\P_{i_n}\cdots\P_{i_1},
\]
where $\P_{i_k}\in \C_k$, $k=1,\dots,n$.
\end{definition}

We will prove that the more general Definition~\ref{def:finer_set_coarse_graining} coincides with this older definition in the limit of projective measurements. In other words, for a class of coarse-grainings $\bC=\{\PO_\bi\}$, $\widetilde{\bC}=\{\PO_j\}$ where
\[
\PO_\bi=\P_{i_n}\cdots\P_{i_1}(\bullet)\P_{i_1}\cdots\P_{i_n}
\]
and
\[
\PO_j=\P_j(\bullet)\P_j,
\]
$\bC\hookleftarrow\widetilde{\bC}$ according to the new definition if and only if $\bC\hookleftarrow\widetilde{\bC}$ according to the old definition.
\begin{proof} We begin by assuming that $\bC\hookleftarrow\widetilde{\bC}$ according to the new definition. This means that for all $j$ there exists an index set $I^{(j)}$ such that
\[
\hPi_j=\sum_{\bi\in I^{(j)}}\hPi_\bi,
\]
where $\hPi_j=\P_j^\dag\P_j=\P_j$ and $\hPi_\bi=\P_\bi^\dag\P_\bi=\P_{i_1}\cdots\P_{i_n}\cdots \P_{i_1}$. Thus, we can rewrite the identity as
\[
\P_j=\sum_{\bi\in I^{(j)}}\P_{i_1}\cdots\P_{i_n}\cdots \P_{i_1}.
\]
Different $\P_j$'s are orthogonal to each other, therefore
\[
\P_{\tilde j}\P_j\P_{\tilde j}=\sum_{\bi\in I^{(j)}}\P_{\tilde j}\P_{i_1}\cdots\P_{i_n}\cdots \P_{i_1}\P_{\tilde j} = 0
\]
for any $\tilde j\neq j$. Applying an arbitrary vector $\ket{\psi}$ on both sides, we get
\[
\sum_{\bi\in I^{(j)}}\norm{\P_{i_n}\cdots \P_{i_1}\P_{\tilde j}\ket{\psi}}^2=0.
\]
Thus, for all $\tilde j\neq j$ and for all $\bi\in I^{(j)}$,
\[
\norm{\P_{i_n}\cdots \P_{i_1}\P_{\tilde j}\ket{\psi}}^2=0.
\]
This holds for any $\ket{\psi}$, therefore also 
\[\label{eq:projectionofP}
\P_{i_n}\cdots \P_{i_1}\P_{\tilde j}=0.
\]
We then pick arbitrary $\bi$ that belongs to some of the index sets $I^{j}$. We, therefore, associate $\P_j$ to this $\bi$ and have
\[
\begin{split}
\P_{i_n}\cdots\P_{i_1}\P_{j}&=\P_{i_n}\cdots\P_{i_1}\big(\P_{j}+\sum_{\tilde j\neq j}\P_{\tilde j}\big)\\
&=\P_{i_n}\cdots\P_{i_1}\I=\P_{i_n}\cdots\P_{i_1},    
\end{split}
\]
where we have used Eq.~\eqref{eq:projectionofP}. This means that $\bC\hookleftarrow\widetilde{\bC}$ according to the old definition. 

Now we prove the opposite implication, that a coarse-graining which satisfies the old definition also satisfies the new definition.
Given $j$, we define $\I^{(j)}$ as the set of all $\bi$ which correspond to $j$ (according to the old definition).
When multiplying Eq.~\eqref{eq:finer_set_condition_old} by $\P_{\tilde j}$, $\tilde j\neq j$, from orthogonality of projectors we obtain
\[
\P_{i_n}\cdots\P_{i_1}\P_{\tilde j}=0
\]
for all $\bi\in \I^{(j)}$, which also implies that 
\[\label{eq:zeroPj}
\P_{i_n}\cdots\P_{i_1}\P_{j}=0
\]
for all $\bi\notin \I^{(j)}$.
Then we have identities
\[
\begin{split}
\hPi_j&=\P_j=\sum_{\bi}\P_{i_1}\cdots\P_{i_n}\cdots \P_{i_1}\P_j\\
&=\sum_{\bi\in I^{(j)}}\P_{i_1}\cdots\P_{i_n}\cdots \P_{i_1}\P_j\\
&=\sum_{\bi\in I^{(j)}}\P_{i_1}\cdots\P_{i_n}\cdots \P_{i_1}=\sum_{\bi\in I^{(j)}}\hPi_\bi. 
\end{split}
\]
where we have used Eqs.~\eqref{eq:zeroPj} and \eqref{eq:finer_set_condition_old}. This implies that $\bC\hookleftarrow\widetilde{\bC}$ according to the new definition, concluding the proof.
\end{proof}

\section{Non-trivialities in generalizing observational entropy to general measurements and key differences}\label{app:nontrivialities}

In this appendix, we point out several non-trivialities encountered when generalizing observational entropy to generalized measurements (POVMs). This is to illustrate the necessary shift in our understanding of this quantity, as well as to motivate using this general definition, which is in many ways elegant and superior to the projective measurement scenario. Below we describe the main generalizations that have made the extension to generalized measurements possible.
\begin{enumerate}
    \item \emph{Volume generalization.--} Original definition of observational entropy relies on the notion of macrostates which are defined as subspaces. The volume of a macrostate is defined as a dimension of the corresponding subspace. This however breaks down when considering multiple, non-commutative projective coarse-grainings. This is because the overlap of two subspaces does not necessarily form a subspace. From an operator perspective, this is connected to the fact that a product of two non-commuting projectors is not a projector. However, the product of two non-commuting projector measurements represents Kraus Rank-1 generalized measurement, which is a form of POVM. Thus, the vector of projective coarse-grainings can be represented by a single POVM, with vector-labeled elements. Also, a vector of POVMs can be represented by a single POVM (see the last line of Table~\ref{tab:typesofmeasurements}). Thus, the composition of projective measurements is not a closed operation (it pushes the definition of coarse-graining outside what is defined as a projective coarse-graining), while the composition of POVMs is a closed operation.

    While for projective measurements, coarse-graining can be represented by a set of operators. For a general POVM, this is no longer possible, and one has to define it as a set of superoperators to include every possible case. When generalizing observational entropy to POVMs, it is clear how to generalize the probabilities of outcomes $p_\bi$. It is apriori not clear how to generalize the corresponding volumes and why one should choose $V_\bi=\tr[\PO_\bi(\I)]$. This is because this quantity is no longer related to any subspace, and therefore $V_\bi$ does not represent a number of microstates contained in that subspace. There can be several other unwanted definitions of the volume. For example, defined by the rank of the corresponding POVM (that would not add up to the total dimension of the system), or the product of local volumes in case of multiple coarse-grainings (which would not lead to the desired properties of observational entropy). However, our choice leads to the theorems to hold and since it is connected to the multiple coarse-graining POVM element instead of the product of single coarse-graning elements, one can intuitively understand why this is a reasonable choice.

    \item \emph{Finer coarse-graining generalization.--} For a projective coarse-graining, there are several definitions of a vector of finer coarse-graining possible~\cite{safranek2019long}. These definitions, while being equivalent for projective coarse-grainings, are however not equivalent when generalized to the POVMs. The particular choice for the Definition~\ref{def:finer_set_coarse_graining} of a finer vector of coarse-grainings is justified by showing that Theorem~\eqref{thm:monotonic} holds for this definition. See Append.~\ref{app:equivalence_definitions} that discusses the equivalence of the definitions for projective coarse-grainings.

    \item \emph{Superiority of Theorem~\ref{thm:monotonic}.--} In the case of projective coarse-grainings, Theorem~\ref{thm:non-increase}  and Theorem~\ref{thm:monotonic} are equivalent statements because one can be derived from the other, when properly rephrased. However, in the case of POVMs, Theorem~\ref{thm:monotonic} is more general, because Theorem~\ref{thm:non-increase} can be derived from it, but not the other way around.
 
\end{enumerate}

\section{Proof of the algorithm for the quantum state inference}\label{app:quantum_state_inference}

Just like in Definition~\ref{def:finer_set_coarse_graining} and elsewhere, we assume that the set of coarse-grainings $\bC=(\C_1,\dots,\C_n)=\{\PO_\bi\}$, where the quantum operations $\PO_{\bi}=\sum_{\bm}\K_{\bi\bm}(\bullet)\K_{\bi\bm}^\dag$ with $\K_{\bi\bm}=\K_{i_nm_n}\cdots\K_{i_1m_1}$, $\bi=(i_1,\dots,i_n)$, and $\bm=(m_1,\dots,m_n)$. Moreover, a POVM element is defined as $\hPi_\bi=\sum_\bm\K_{\bi\bm}^\dag\K_{\bi\bm}$.

According to the Definition~\ref{def:finer_set_coarse_graining}, $\bC\hookleftarrow \C_{\R}=\{\P_\rho(\bullet)\P_\rho\}$ if and only if we can build each eigenprojector of the density matrix using POVM elements from the vector of coarse-graining $\bC$, i.e., if for each eigenvalue $\rho$ there exists an index set $I^{(\rho)}$ such that
\[
\P_\rho=\sum_{\bi\in I^{(\rho)}}\hPi_{\bi}=\sum_{\bi\in I^{(\rho)},\bm}\K_{\bi\bm}^\dag\K_{\bi\bm}.
\]
For the following proofs needed for the algorithm, we need to identify how to group the POVM elements together so that we can build those projectors, i.e., we need an algorithm of how to generate these sets $I^{(\rho)}$ that can be used to build up $\P_\rho$'s. This will be achieved by the following three lemmas:

\begin{lemma}\label{lma:corresponding_to_same}
Let $\bC\hookleftarrow \C$, where $\C=\{\P_j(\bullet)\P_j\}$, and $\bC=\{\PO_\bi\}$. If two POVM elements have non-zero overlap, they must correspond to the same projector $\P_j$. In mathematical terms, if $\hPi_\bi\hPi_{\bi'}\neq 0$, then both $\bi,\bi'\in I^{(j)}$ for some $j$.
\end{lemma}
\begin{proof}
For contradiction, let us assume that $\bi\in I^{(j)}$ and $\bi'\in I^{(\tilde j)}$, where $\tilde j\neq j$. Then from the orthogonality of projectors we have
\begin{eqnarray}
\P_{j}\P_{\tilde j}\P_{j} &=& 0,\nonumber \\
\P_{j}\sum_{\bi'\in I^{(\tilde j)},\bm}\K_{\bi'\bm'}^\dag\K_{\bi'\bm'}\P_{j} &=& 0.
\end{eqnarray}
Applying an arbitrary vector $\ket{\psi}$ from both sides we get
\[
\sum_{\bi'\in I^{(\tilde j)},\bm'}\norm{\K_{\bi'\bm'}\P_{j}\ket{\psi}}^2 = 0.
\]
This holds for any $\ket{\psi}$, therefore
\[
\K_{\bi'\bm'}\P_{j}=0 \quad \quad \quad (\forall \bi'\in I^{(\tilde j)} ~\& ~\forall \bm').
\]
Multiplying this equation by $\K_{\bi'\bm'}^\dag$ from the right, expressing $\P_j$ in terms of the Krauss operators, and applying $\ket{\psi}$ from both sides we get
\[
\sum_{\bi\in I^{(j)},\bm}\norm{\K_{\bi\bm}\K_{\bi'\bm'}^\dag\ket{\psi}}^2 = 0.
\]
Since this holds for any $\ket{\psi}$, we get
\[
\K_{\bi\bm}\K_{\bi'\bm'}^\dag=0 ~ (\forall \bi \in I^{(j)},\forall \bi'\in I^{(\tilde j)},\forall \bm,~\&~\forall \bm').
\]
Multiplying this equation by $\K_{\bi\bm}^\dag$ from the left and by $\K_{\bi'\bm'}$ from the right and summing over $\bm$ and $\bm'$ we get
\[
\hPi_\bi\hPi_{\bi'}=0,
\]
which is in contradiction with our assumption that $\hPi_\bi\hPi_{\bi'}\neq0$. Therefore, both $\bi$ and $\bi'$ must belong into the same $I^{(j)}$.
\end{proof}

\noindent {\bf Lemma~\ref{lma:eigenprojector}.} Next we go on to prove Lemma~\ref{lma:eigenprojector} that says $\P$'s generated by steps 1--3 form a complete set of orthogonal projectors, and that each $\P$ projects into an eigenspace of the state of the system.
\begin{proof} Lemma~\ref{lma:corresponding_to_same} and Definition~\ref{def:finer_set_coarse_graining} guarantees that the procedure in step 3 will correctly generate the elements of the index set $I^{(\rho)}$. This is because from Lemma~\ref{lma:corresponding_to_same} we know that all POVM elements that have non-zero overlap must correspond to the same projector $\P_\rho$. Let us denote the set of all operators $\P$ generated through the algorithm as $A=\{\P\}$. This set must be complete ($\sum_{\P\in A} \P=\I$), because $\sum_{\P\in A} \P=\sum_\bi \hPi_\bi=\sum_{\bi,\bm}\K_{\bi\bm}^\dag\K_{\bi\bm}=\I$. Combining Lemma~\ref{lma:corresponding_to_same} and Definition~\ref{def:finer_set_coarse_graining} we know that for each $\rho$ there must exist a subset of $A$, let us denote it $B^{\rho}$, such that
\[\label{eq:Prho_expressed_with_P}
\P_\rho=\sum_{\P\in B^{\rho}}\P.
\]
Further, we know that every operator $\P\in A$ must be orthogonal to every other operator $\P\in A$. This is because the iterative process of step 3 is stopped exactly when for all unused $\bi'$, $\hPi_{\bi'}\P=0$, so $\P$ must be orthogonal to all $\P$'s that will come after it, and by the same logic, $\P$ must also be orthogonal to all the $\P$'s that came before it. Also, from the construction it is clear that $B^{\rho}$'s are disjoint sets ($B^{\rho}\cap B^{\rho'}=\emptyset$ for $\rho\neq\rho'$) whose union equals the set A ($A=\bigcup_\rho B^{\rho}$). If we manage to prove that each $\P$ is a projector, then from Eq.~\eqref{eq:Prho_expressed_with_P} it is clear that it must project into an eigenspace of the density matrix, simply because $\P_\rho$ already does by definition.

Therefore, all we have to do now is to prove that each $\P$ generated by step 3 is a projector. For a contradiction, let us assume that there is some $\P$ that is not a projector. This $\P$ belongs into some set $B^{\rho}$. To make the notation clearer, let us label elements of $B^{\rho}$ as $\P_1,\dots,\P_k$ so we can write
\[
\P_\rho=\P_1+\cdots+\P_k,
\]
and we assume that operator $\P_1$ is the one that is not a projector. $\P_\rho$ is a projector, therefore
\begin{eqnarray}
\P_\rho=\P_\rho^2 &=& (\P_1+\cdots+\P_k)^2\nonumber \\
&=&\P_1^2+\cdots+\P_k^2,  
\end{eqnarray}
where we have used that all $\P$'s are orthogonal to each other. Multiplying this equation by $\P_1$ and using the orthogonality again we obtain
\[
\P_1^2=\P_1^3.
\]
However, $\P_1$ is a Hermitian operator by construction, so there is a spectral decomposition of $\P_1=\sum_n\lambda_n\pro{\psi_n}{\psi_n}$, $\lambda_n\neq 0$, where $\ket{\psi_n}$ are orthogonal. Combining the spectral decomposition with the above equation, we have
\[
\sum_n\lambda_n^2\pro{\psi_n}{\psi_n}=\sum_n\lambda_n^3\pro{\psi_n}{\psi_n}.
\]
Due to orthogonality of $\ket{\psi_n}$, it must be that $\lambda_n^2=\lambda_n^3$ for every $n$, in other words, $\lambda_n=1$ for every $n$. Thus,
\[
\P_1=\sum_n\pro{\psi_n}{\psi_n},
\]
therefore $\P_1$ is a projector, which is a contradiction. Thus, $\P$ generated by step 3 is always a projector, and as we established before, it must project into an eigenspace of the density matrix. This concludes the proof.
\end{proof}

\noindent{\bf Lemma~\ref{lma:eigenvalue}.} Finally, we prove Lemma~\ref{lma:eigenvalue} which says that $\rho$ generated by step 4 is an eigenvalue of the state of the system corresponding to projector $\P$.
\begin{proof}From Lemma~\ref{lma:eigenprojector} we know that each $\P$ projects into an eigenspace of the system. Let us denote the (currently unknown) eigenvalue corresponding to this eigenspace as $\rho_0$. $\P$ is orthogonal to every other $\P$ by construction $\P=\sum_{\bi\in I}\hPi_\bi$. The probability of obtaining measurement outcome $\bi$ is equal to
\[
p_\bi=\tr[\PO_\bi(\R)]=\tr[\hPi_\bi\R]=\tr[\hPi_\bi \sum_\rho \rho \P_\rho].
\]
Summing over all $\bi\in I$ we obtain
\[
\sum_{\bi\in I}p_\bi=\tr[\P \sum_\rho \rho\P_\rho]=\tr[\sum_\rho \rho\,\delta_{\rho \rho_0}\P]=\rho_0\tr[\P].
\]
$I$ is a non-empty set by construction, therefore $\tr[\P]$ is non-zero. Thus we can divide by it and obtain
\[
\rho_0=\frac{\sum_{\bi\in I}p_\bi}{\tr[\P]},
\]
which concludes the proof.
\end{proof}

\section{Explicit form of repeated measurements and repeated contacts superoperators}\label{app:explicit_forms}

\subsection{Repeated measurements}
It can be easily realized that the repeated measurements superoperator~\eqref{eq:risuper} can be rewritten as 
\[
\begin{split}
\PO_\bx^{\mathrm{rm}}(\hat{Z})=\K_\bx^{\mathrm{rm}}\hat{Z}\K_\bx^{\mathrm{rm}}
\end{split}
\]
where
\[
\K_\bx^{\mathrm{rm}}=\U_f\K_{x_N}\cdots \U_f\K_{x_1}.
\]
In terms of the bare elements, this gives
\[\label{eq:risuperapp}
(\K_\bx^{\mathrm{rm}})_{m_Nm_0}=\sum_{m_1,\dots,m_{N-1}}\U_{m_0\dots m_N}\sqrt{g_{m_0\dots m_N}^{\mathrm{rm}}(\bx)}
\]
where the contribution due to the free-evolution is given by 
\[
\U_{m_0\dots m_N}=(\U_{f})_{m_Nm_{N-1}}\dots (\U_{f})_{m_1m_0}.
\]
The Gaussian state of the pointer modifies as 
\[
\sqrt{g_{m_0\dots m_N}^{\mathrm{rm}}\!(\bx)}\!=\!\frac{1}{(2\pi \Omega^2)^{1/4}}\exp\bigg[\!-\frac{(\bx\!-\!\kappa\boldsymbol{\mu}_{m_0\dots m_N})^2}{4\Omega^2}\!\bigg],
\]
where $\bx=(x_1,\dots,x_N,0)$ and $\boldsymbol{\mu}_{m_0\dots m_N}=(\mu_{m_0},\dots,\mu_{m_{N-1}},0)$.

\subsection{Repeated contacts}
To obtain an explicit form of the repeated contacts superoperator~\eqref{eq:rcsuper}, we first compute
\begin{widetext}
\[
\begin{split}
(\mathcal{U}_f\mathcal{U})^N(\hat{Z}\otimes\S)&=\!\!\!\!\sum_{m_0,m_0',\dots,m_N,m_N'}\!\!\!\!\!\!\!\!(\U_{f})_{m_Nm_{N-1}}\dots (\U_{f})_{m_1m_0}\hat{Z}_{m_0m_0'}(\U_{f})_{m_N'm_{N-1}'}^*\dots (\U_{f})_{m_1'm_0'}^*\\
&\times\pro{m_N}{m_N'}\otimes \pro{\varphi_{x-\kappa\mu_{m_0}\cdots-\kappa\mu_{m_{N-1}}}}{\varphi_{x-\kappa\mu_{m_0'}\cdots-\kappa\mu_{m_{N-1}'}}}.
\end{split}
\]
\end{widetext}
This leads to
\[
\begin{split}
\PO_x^{\mathrm{rc}}(\hat{Z})
&=\bra{x}_B(\mathcal{U}_f\mathcal{U})^N(\hat{Z}\otimes\S)\ket{x}_B \\
&=\K^{\mathrm{rc}}_x\hat{Z}\K^{\mathrm{rc\dag}}_x,
\end{split}
\]
where
\begin{eqnarray}\label{eq:computation_repeated_contacts}
\K^{\mathrm{rc}}_x&=&\sum_{m_0,\dots,m_N}(\U_{f})_{m_Nm_{N-1}}\dots (\U_{f})_{m_1m_0}  \\
&&\times \sqrt{g^\Omega_{x-\kappa\mu_{m_0}\cdots-\kappa\mu_{m_{N-1}}}}\pro{m_N}{m_0},\nonumber
\end{eqnarray}
which represents the elements of the operator $\PO_x^{\mathrm{rc}}(\hat{Z})$ in the eigenbasis of the measurement operator.
We can write this in a form similar to Eq.~\eqref{eq:risuperapp}
\[\label{eq:ri_K}
(\K_x^{\mathrm{rc}})_{m_Nm_0}=\sum_{m_1,\dots,m_{N-1}}\U_{m_0\dots m_N}\sqrt{g_{m_0\dots m_N}^{\mathrm{rc}}(x)}
\]
where
\begin{eqnarray}\label{eq:evaluating_gaussian}
\sqrt{g_{m_0\dots m_N}^{\mathrm{rc}}(x)}&=&\frac{1}{(2\pi \Omega^2)^{1/4}} \\
&\times& \exp\bigg[-\frac{(x-\kappa \sum_{i=0}^{N-1}\mu_{m_{i}})^2}{4\Omega^2}\bigg].\nonumber
\end{eqnarray}
Comparing Eqs.~\eqref{eq:risuperapp} and~\eqref{eq:ri_K} we see that the only difference between the repeated measurements and repeated contacts is the function $\sqrt{g_{m_0\dots m_N}^{\mathrm{rm}}(\bx)}$ and $\sqrt{g_{m_0\dots m_N}^{\mathrm{rc}}(x)}$, respectively.

\section{Numerical details}\label{app:numerics}

On a computer we cannot create an exactly continuous function, hence in all figures, we choose the coarse-grained positional step to be $dx=0.1$. Since computing observational entropy is a numerical integration problem, using Trapezoidal rule~\cite{trefethen2014exponentially} we expect the error to scale as $dx^2$, which would give precision around $\mathcal{O}(10^{-2})$. However, in practice, we saw that observational entropy is very insensitive to the coarse-graining size in $x$, and the error is much smaller. %{\bf I recall you doing a calculation for a smaller $dx$, please state the value here.}. 
Notably, the difference in results between $dx=1$ and $dx=0.1$ was only $10^{-3}$ or less in every example we considered, and the difference in outcome for $dx=0.1$ and $dx=0.01$ was of order $\mathcal{O}(10^{-6})$ for cases $N=2$. We also chose a cut-off in the positional axis at four standard deviations $\Omega$ from the furthest peak of Gaussians that make $p$'s and $V$'s (for instance, see Fig.~\ref{Fig:non} (b) and (c) that show the cut-off scale). We performed a series of sanity checks, for example, summing computed probabilities $p$ and volumes $V$ and making sure they add up to $1$ and $2=\dim \HS_S$ respectively and ensuring that the generated coarse-graining sums to unity as per Eq.~\eqref{eq:completeness}.

Computing the observational entropy for the von Neumann measurement schemes is a computationally demanding task. This is due to the dimensionality of the problem growing exponentially with $N$ and the sheer number of points for which we need to compute $p$ and $V$, especially in the repeated measurement case. 

In repeated contacts, we can estimate the computational complexity as the product of three critical factors: i) the number of discretized points on the position axis (computed via the distance between the furthest Gaussian peaks + 4$\Omega$ buffer on both sides), i.e., $(\kappa N (\mu_{\max}-\mu_{\min})+ 8\Omega)/dx$ , ii) complexity of computing a single element of the sum [Eq.~\eqref{eq:computation_repeated_contacts}], for a given $x$, has a leading order $N^2$ for having to multiply $N$ elements of $(\U_{f})_{m_im_{i-1}}$, and iii) summing all $(\dim \HS_S)^N$ elements of the sum~\eqref{eq:computation_repeated_contacts}. Thus, the leading order of the computational complexity $(\kappa N (\mu_{\max}-\mu_{\min})+8\Omega)  N^2 (\dim \HS_S)^N/dx$ scales with the number of contacts $N$ as $\mathcal{O}(N^3 (\dim \HS_S)^N)$. With our parameters ($\mu_{\max}=2$, $\mu_{\min}=0$, $\Omega=1$, $\dim \HS_S=2$), we obtain $N^32^N$ as the leading order. This translates into the ability to compute for up to $N=22$ on a single processor with current parameters.

In repeated measurements, we can estimate the computational complexity as the product of three critical factors: i) the number of discretized points $\bx=(x_1,\dots,x_N)$ on the position axes (computed via the distance between the furthest Gaussian peaks + 4$\Omega$ buffer on both sides, to the power of the number of different axes), i.e., $\Big((\kappa (\mu_{\max}-\mu_{\min})+ 8\Omega)/dx\big)^N$, ii) complexity of computing a single element of the sum [Eq.~\eqref{eq:risuperapp}], for a given $\bx$, has a leading order $N^2$ for having to multiply $N$ elements $(\U_{f})_{m_im_{i-1}}$, and iii) summing all $(\dim \HS_S)^N$ elements of the sum~\eqref{eq:risuperapp}. %, which is {\bf state here what gives rise to $(\dim \HS_S)^N$ term. Moreover, forget about the additional operations since they do not contribute to computational complexity on leading order.}. 
Thus, the leading order of the computational complexity $(\kappa (\mu_{\max}-\mu_{\min})/dx+8\Omega)^N  N^2 (\dim \HS_S)^N$ scales with the number of contacts $N$ as $\mathcal{O}\big(N^2 \Big(\dim \HS_S(\kappa (\mu_{\max}-\mu_{\min})+ 8\Omega)/dx\big)^N\big)$. With our parameters ($\mu_{\max}=2$, $\mu_{\min}=0$, $\Omega=1$, $\dim \HS_S=2$), we obtain $N^2 200^N$ as the leading order. This translates into the ability to compute for up to $N=5$ on a single processor with current parameters. The task is also highly parallelizable, however, any advantage gained by parallelization diminishes quickly with $N$ due to the prohibitive scaling.

\bibliographystyle{apsrev4-1}
\bibliography{main.bib}

\end{document}